\documentclass[11pt]{article}
\usepackage{graphicx}
\usepackage{tabularx}
\usepackage{url}
\usepackage{amsthm}
\usepackage{amsmath}
\usepackage{amsfonts}

\theoremstyle{plain}
\newtheorem{theorem}{Theorem}
\newtheorem{lemma}{Lemma}
\newtheorem{proposition}{Proposition}

\theoremstyle{definition}

\theoremstyle{remark}

\date{}

\begin{document}

\title{Time and space complexity of deterministic and nondeterministic decision trees}
\author{Mikhail Moshkov\thanks{Computer, Electrical and Mathematical Sciences and Engineering Division,
King Abdullah University of Science and Technology (KAUST),
Thuwal 23955-6900, Saudi Arabia. Email: mikhail.moshkov@kaust.edu.sa.
}}
\maketitle

\begin{abstract}
In this paper, we study arbitrary infinite binary information systems each of which consists of an infinite set called universe and an infinite set of two-valued functions (attributes) defined on the universe. We consider the notion of a problem over information system which is described by a finite number of attributes and a mapping corresponding a decision to each tuple of attribute values. As algorithms for problem solving, we use deterministic and nondeterministic decision trees. As time and space complexity, we study the depth and the number of nodes in the decision trees. In the worst case, with the growth of the number of attributes in the problem description, (i) the minimum depth of deterministic decision trees grows either almost as logarithm or linearly, (ii) the minimum depth of nondeterministic decision trees either is bounded from above by a constant or grows linearly, (iii) the minimum number of nodes in deterministic decision trees has either  polynomial or exponential growth, and (iv) the minimum number of nodes in nondeterministic decision trees has either  polynomial or exponential growth. Based on these results, we divide the set of all infinite binary information systems  into five complexity classes, and study for each class issues related to time-space trade-off for decision trees.
\end{abstract}

{\it Keywords}: deterministic decision trees, nondeterministic decision trees, time complexity, space complexity, complexity classes, time-space trade-off.

\section{Introduction} \label{S1}

In this paper, we divide the set of all infinite binary information systems  into five complexity classes depending on the worst case time and space complexity of deterministic an nondeterministic decision trees, and study for each class issues related to time-space trade-off for decision trees.

General information system \cite{Pawlak81}  consists of
a universe (a set of objects) and a set of
attributes (functions with finite image) defined on the universe. An information system is called infinite,
if its set of attributes is infinite. Otherwise, it is called finite. An information system is called binary if each its attribute has values from the set $\{0,1\}$.

Any problem over an information system is described by
a finite number of attributes that divide the universe into domains in which these
attributes have fixed values. A decision is attached to each domain. For a
given object from the universe, it is required to find the decision attached
to the domain containing this object.

As algorithms solving these problems,
deterministic and nondeterministic decision trees are considered.
Deterministic decision trees
are widely used as classifiers to predict decisions for new objects, as a
means of knowledge representation, and as algorithms to solve problems of fault diagnosis,
computational geometry, combinatorial optimization, etc. \cite{Breiman84,Moshkov05,Rokach07}.
Nondeterministic decision trees are less known. They are closely
related to systems of true decision rules that cover all objects from the universe.
As time complexity of a decision tree, we consider its depth -- the maximum number of nodes labeled with attributes in a path from the root to a terminal node. As space complexity of a decision tree, we consider its number of nodes.

Both theoretical and experimental investigations of time complexity of nondeterministic decision trees are mainly related to decision trees for Boolean functions \cite{AbouEisha19,Blum87,Moshkov95a}. Note that, for a Boolean function, the minimum depth of a nondeterministic decision tree is equal to its certificate complexity \cite{Buhrman02}.

The most part of results on deterministic decision trees is obtained for finite information systems.
The results related to infinite information systems were achieved initially in the study of deterministic linear
and algebraic decision trees. In this case, the universe is a subset of $n$-dimensional real space and each attribute is of the kind $\text{sign} f$, where  $f$ is a linear form with $n$ variables
for linear decision trees, and  $f$ is a polynomial with $n$ variables for algebraic
decision trees.

In \cite{Dobkin78,Dobkin79,Moravek72}, the lower bounds close to $n \log_2t$ were obtained
for the minimum depth of linear decision trees, where $t$ is the number of attributes in the problem description. Lower bounds on the minimum depth of algebraic
decision trees  were obtained later \cite{Ben-Or83,Grigoriev95,Steele82,Yao92,Yao94} as well as
lower bounds on the minimum number of nodes in the algebraic decision trees \cite{Grigoriev98}.

In \cite{Dobkin76}, the upper bound
$(3 \cdot 2^{n-2}+n-2)(\log_2 t +1)$
for the minimum depth
of linear decision trees was obtained  for $n \ge 2$. The paper \cite{Moshkov82} contains the upper bound
$(2(n + 2)^3 \log_2(t + 2n+ 2))/ \log_2(n + 2)$.  Similar upper
bound was obtained in \cite{Heide83}.

In our view, the problems of complexity of decision trees over arbitrary infinite information systems were not considered prior to \cite{Moshkov94a,Moshkov94b} for deterministic and prior to \cite{Moshkov95,Moshkov96} for nondeterministic decision trees.

We developed two approaches to the study of deterministic and nondeterministic decision trees over arbitrary information systems: local, when the decision trees solving a problem can use only attributes from the problem description, and global, when the decision trees solving a problem can use arbitrary attributes from the considered information system \cite{Alsolami20,Moshkov94b,Moshkov96,Moshkov03,Moshkov05a,Moshkov05,Moshkov20,Moshkov11}.

Based on the obtained results we can describe possible types of behavior of four functions $h_U^d,h_U^a,L_U^d,L_U^a$ that characterize worst case time and space complexity of deterministic and nondeterministic decision trees over an infinite binary information system $U$.

The function $h_U^d$ characterizes the growth in the worst case of the minimum depth of a deterministic decision tree solving a problem with the growth of the number of attributes in the problem description. The function $h_U^d$ is either bounded from below by logarithm and bounded from above by logarithm to the power $1+\varepsilon$, where $\varepsilon$ is an arbitrary positive real number, or grows linearly.

The function $h_U^a$ characterizes the growth in the worst case of the minimum depth of a nondeterministic decision tree solving a problem with the growth of the number of attributes in the problem description. The function $h_U^a$ is either bounded from above by a constant or grows linearly.

The function $L_U^d$ characterizes the growth in the worst case of the minimum number of nodes in a deterministic decision tree solving a problem with the growth of the number of attributes in the problem description. The function $L_U^d$ has either polynomial or exponential growth.

The function $L_U^a$ characterizes the growth in the worst case of the minimum number of nodes in a nondeterministic decision tree solving a problem with the growth of the number of attributes in the problem description. The function $L_U^a$ has either polynomial or exponential growth.

We see that each of the functions $h_U^d,h_U^a,L_U^d,L_U^a$ has two types of behavior. Thus,  the tuple $(h_U^d,h_U^a,L_U^d,L_U^a)$ can have (a priori) 16 types of behavior. However (and this is one of the main results of the paper), the tuple $(h_U^d,h_U^a,L_U^d,L_U^a)$ can have only five types of behavior. All these types are enumerated in the paper and each type is illustrated by an example. Similar result without proofs and with weaker bounds on the functions $h_U^d,h_U^a,L_U^d,L_U^a$ was announced in  \cite{Moshkov00}.

There are five complexity classes of infinite binary information systems corresponding to the five possible types of  the tuple $(h_U^d,h_U^a,L_U^d,L_U^a)$. For each class, we study joint behavior of time and space complexity of decision trees.

A pair of functions $(\varphi ,\psi )$ is called a boundary $d$-pair of the information system $U$ if, for any problem over $U$, there exists a deterministic decision tree over $U$ which solves this problem and for which the depth is at most $\varphi (n)$ and the number of nodes is at most $ \psi (n)$, where $n$ is the number of attributes in the problem description. A boundary $d$-pair $(\varphi ,\psi )$ of the information system $U$ is called optimal if, for any boundary $d$-pair $(\varphi ^{\prime },\psi ^{\prime })$ of $U$, the inequalities $\varphi ^{\prime }(n)\geq \varphi (n)$ and $\psi ^{\prime}(n)\geq \psi (n)$ hold for any natural $n$. An information system $U$ is called $d$-reachable if the pair $(h_{U}^{d},L_{U}^{d})$ is boundary (and, consequently, optimal boundary) $d$-pair of the system $U$. For nondeterministic decision trees, the notions of a boundary $a$-pair of an information system, an optimal $a$-pair, and $a$-reachable information system are  defined in a similar way. For deterministic decision trees, the best situation is when the considered information system is $d$-reachable: for any boundary $d$-pair $(\varphi ,\psi )$ for an information system $U$ and any natural $n$, $\varphi (n) \ge h_{U}^{d}(n)$ and  $\psi (n) \ge L_{U}^{d}(n)$.
For nondeterministic decision trees, the best situation is when the information system is $a$-reachable.

For four out of the five complexity classes, all information systems from the class are $d$-reachable. One class contains both information systems that are $d$-reachable and information systems that are not $d$-reachable. For each information system $U$ that is not $d$-reachable, we find a nontrivial boundary $d$-pair which is enough close to the pair $(h_{U}^{d},L_{U}^{d})$.
For two out of the five complexity classes, all information systems from the class are $a$-reachable. For the rest three classes, all information systems from the class are not $a$-reachable. For some information systems $U$ that are not $a$-reachable, we find nontrivial boundary $a$-pairs which are enough close to  $(h_{U}^{a},L_{U}^{a})$. For the rest of information systems $U$ that are not $a$-reachable, the pair $(n,L_{U}^{a}(n))$ is the optimal boundary $a$-pair. Note that for these information systems, the function $h_{U}^{a}$ is bounded from above by a constant.

The obtained results are related to time-space trade-off for deterministic and nondeterministic decision trees. For any information system $U$ for each problem, there exists a deterministic decision tree solving this problem which depth  is at most $h_{U}^{d}(n)$, and there exists a deterministic decision tree solving this problem for which the number of nodes is at most $L_{U}^{d}(n)$, where $n$ is the number of attributes in the problem description.
If an information system $U$ is not $d$-reachable, then there exists a problem such that there is no a deterministic decision tree solving this problem which depth  is at most $h_{U}^{d}(n)$ and the number of nodes is at most $L_{U}^{d}(n)$, where $n$ is the number of attributes in the problem description. Similar situation is with nondeterministic decision trees for information systems that are not $a$-reachable.

Let us consider an information system $U$ for which the function $h_{U}^{a}$ is bounded from above by a natural number $c$, and $(n,L_{U}^{a}(n))$ is the optimal boundary $a$-pair. For each problem over $U$, there exists a nondeterministic decision tree solving this problem which depth is at most $c$. However,  for any natural $n$ greater than $c$, there is no a finite upper bound on the number of nodes in such trees for problems described by at most $n$ attributes.

Note that a part of the obtained results can be extended to infinite $k$-valued information systems, $k > 2$, in particular, the results about five possible types of infinite binary information systems -- see \cite{Moshkov00}.

The rest of the paper is organized as follows: Section \ref{S2} contains main results and Sections \ref{S3}-\ref{S5} -- proofs of these results.
\section{Main Results} \label{S2}

Let $A$ be an infinite set and $F$ be an infinite set of functions that are
defined on $A$ and have values from the set $\{0,1\}$. The pair $U=(A,F)$ is
called an infinite binary information system \cite{Pawlak81}, elements of the set $A$
are called objects, and functions from $F$ are called attributes. The set $A$ is called sometimes the universe of the information system $U$.

A problem over $U$ is a tuple of the kind $z=(\nu ,f_{1},\ldots ,f_{n})$,
where $\nu :\{0,1\}^n\rightarrow \mathbb{N}$, $\mathbb{N}$ is the set of
natural numbers $\{1,2,\ldots \}$, and $f_{1},\ldots ,f_{n}\in F$. The
problem $z$ consists in finding the value of the function  $z(x)=\nu
(f_{1}(x),\ldots ,f_{n}(x))$ for a given object $a\in A$. Various problems
of pattern recognition, combinatorial optimization, fault diagnosis,
computational geometry, etc., can be represented in this form. The value $\dim
z=n $ is called the dimension of the problem $z$.

As algorithms for problem solving we consider decision trees. A decision
tree over the information system $U$ is a directed tree with the root in
which the root and edges leaving the root are not labeled, each terminal
node is labeled with a number from $\mathbb{N}$, each working node (which is
neither the root nor a terminal node) is labeled with an attribute from $F$,
and each edge leaving a working node is labeled with a number from the set $%
\{0,1\}$. A decision tree is called deterministic if only one edge leaves
the root and edges leaving an arbitrary working node are labeled with
different numbers.

Let $\Gamma $ be a decision tree over $U$ and $$\xi
=v_{0},d_{0},v_{1},d_{1},\ldots ,v_{m},d_{m},v_{m+1}$$ be a directed path
from the root $v_0$ to a terminal node $v_{m+1}$ of $\Gamma $ (we call such path complete).
Define a subset $A(\xi )$ of the set $A$. If $m=0$, then $A(\xi )=A$. Let $%
m>0$ and, for $i=1,\ldots ,m$, the node $v_{i}$ be labeled with the
attribute $f_{j_{i}}$ and the edge $d_{i}$ be labeled with the number $%
\delta _{i}$. Then $$A(\xi )=\{a:a\in A,f_{j_{1}}(a)=\delta _{1},\ldots
,f_{j_{m}}(a)=\delta _{m}\}.$$

The decision tree $\Gamma $ solves the problem $z$ nondeterministically if,
for any object $a\in A$, there exists a complete path $\xi $ of $\Gamma $
such that $a\in A(\xi )$ and, for each $a\in A$ and each complete path $\xi $
such that $a\in A(\xi )$, the terminal node of $\xi $ is labeled with the
number $z(a)$ (in this case, we can say that $\Gamma$ is a nondeterministic decision tree solving the problem $z$).
In particular, if the decision tree $\Gamma $ solves the
problem $z$ nondeterministically, then, for each complete path $\xi $ of $%
\Gamma $, either the set $A(\xi )$ is empty or the function $z(x)$ is constant
on the set $A(\xi )$. The decision tree $\Gamma $ solves the problem $z$
deterministically if $\Gamma $ is a deterministic decision tree which solves
the problem $z$ nondeterministically (in this case, we can say that $\Gamma$ is a deterministic decision tree solving the problem $z$).

The depth of the decision tree $\Gamma $ is the maximum number of working nodes
in a complete path of $\Gamma $. Denote $h(\Gamma )$ the depth of $\Gamma $
and $L(\Gamma )$ -- the number of nodes in $\Gamma $.

Let $P(U)$ be the set of problems over $U$. For a problem $z$ from $P(U)$,
let $h_{U}^{d}(z)$ be the minimum depth of a decision tree over $U$ solving
the problem $z$ deterministically, $h_{U}^{a}(z)$ be the minimum depth of a
decision tree over $U$ solving the problem $z$ nondeterministically, $%
L_{U}^{d}(z)$ be the minimum number of nodes in a decision tree over $U$
solving the problem $z$ deterministically, and $L_{U}^{a}(z)$ be the minimum
number of nodes in a decision tree over $U$ solving the problem $z$
nondeterministically.

We consider four functions defined on the set $\mathbb{N}$ in the following way: $%
h_{U}^{d}(n)=\max $ $h_{U}^{d}(z)$, $h_{U}^{a}(n)=\max $ $h_{U}^{a}(z)$, $%
L_{U}^{d}(n)=\max $ $L_{U}^{d}(z)$, and $L_{U}^{a}(n)=\max $ $L_{U}^{a}(z)$,
where the maximum is taken among all problems $z$ over $U$ with $\dim z\leq
n $. These functions describe how the minimum depth and the minimum number
of nodes of deterministic and nondeterministic decision trees solving
problems are growing in the worst case with the growth of problem dimension.
To describe possible types of behavior of these four functions, we need to
define some properties of infinite binary information systems.

Let $m\in \mathbb{N}$. A nonempty subset $B$ of the set $A$ is called a  $(m,U)$%
-set if $B$ coincides with the set of solutions from $A$ of an equation
system of the kind %
$
\{f_{1}(x)=\delta _{1},\ldots ,f_{m}(x)=\delta _{m}\},
$
where $f_{1},\ldots ,f_{m}$ are attributes from $F$ (not
necessary pairwise different), and $\delta _{1},\ldots ,\delta _{m}\in
\{0,1\}$. We call such system an $(m,U)$-system of equations. It is clear that an $(m,U)$-set is also an $(m+1,U)$-set.

We say that the information system $U$ satisfies the condition of coverage
if there exists $m\in \mathbb{N}$ such that any $(m+1,U)$-set is a union of
a finite number of $(m,U)$-sets. In this case, we will say that $U$
satisfies the condition of coverage with parameter $m$.

We say that the information system $U$ satisfies the condition of restricted
coverage if there exist $m,t\in \mathbb{N}$ such that any $(m+1,U)$-set is a
union of at most $t$ $(m,U)$-sets. In this case, we will say that $U$
satisfies the condition of restricted coverage with parameters $m$ and $t$.

A subset $\{f_{1},\ldots ,f_{m}\}$ of the set $F$ is called independent if,
for any $\delta _{1},\ldots ,\delta _{m}\in \{0,1\}$, the system of
equations $
\{f_{1}(x)=\delta _{1},\ldots ,f_{m}(x)=\delta _{m}\},
$ has a solution from the set $A$. The empty set of
attributes is independent by definition. We define the parameter $I(U)$
which is called the independence dimension or I-dimension of the
information system $U$ (this notion is similar to the notion of independence number of family of sets \cite{Naiman96}).
If, for each $m\in \mathbb{N}$, the set $F$ contains an independent subset
of cardinality $m$, then $I(U)=\infty $. Otherwise, $I(U)$ is the maximum
cardinality of an independent subset of the set $F$.

We now consider four statements that describe possible types of behavior of
functions $h_{U}^{d}(n)$, $h_{U}^{a}(n)$, $L_{U}^{d}(n)$, and $L_{U}^{a}(n)$%
. The next statement follows immediately from Theorem 2.1 \cite{Moshkov03} and simple fact
that $h_{U}^{d}(n)\leq n$ for any $n\in \mathbb{N}$.

\begin{proposition}
\label{P1}For any infinite binary information system $U$, the function $%
h_{U}^{d}(n)$ has one of the following two types of behavior:

 {\rm (LOG)} If the system $U$ has finite I-dimension and satisfies the
condition of restricted coverage, then for any $\varepsilon $, $%
0<\varepsilon <1$, there exists a positive constant $c$ such that, for any $%
n\in \mathbb{N}$,
\[
\log _{2}(n+1)\leq h_{U}^{d}(n)\leq c(\log _{2}n)^{1+\varepsilon }+1.
\]

 {\rm (LIN)} If the system $U$ has infinite I-dimension or does not
satisfy the condition of restricted coverage, then for any $n\in \mathbb{N}$%
,
\[
h_{U}^{d}(n)=n.
\]
\end{proposition}

\begin{proposition}
\label{P2}For any infinite binary information system $U=(A,F)\,$, the
function $h_{U}^{a}(n)$ has one of the following two types of behavior:

 {\rm (CON)} If the system $U$ satisfies the condition of coverage, then
there exists a positive constant $c$ such that, for any $n\in \mathbb{N}$,
\[
h_{U}^{a}(n)\leq c.
\]

{\rm (LIN)} If the system $U$ does not satisfy the condition of coverage,
then for any $n\in \mathbb{N}$,
\[
h_{U}^{a}(n)=n.
\]
\end{proposition}

\begin{proposition}
\label{P3}For any infinite binary information system $U$, the function $%
L_{U}^{d}(n)$ has one of the following two types of behavior:

{\rm (POL)} If the system $U$ has finite I-dimension, then for any $%
n\in \mathbb{N}$,
\[
2(n+1)\leq L_{U}^{d}(n)\leq 2(4n)^{I(U)}.
\]

{\rm (EXP)} If the system $U$ has infinite I-dimension, then for any $%
n\in \mathbb{N}$,
\[
L_{U}^{d}(n)=2^{n+1}.
\]
\end{proposition}

\begin{proposition}
\label{P4}For any infinite binary information system $U$ and any $n\in
\mathbb{N}$,%
\[
L_{U}^{a}(n)=L_{U}^{d}(n).
\]
\end{proposition}

Let $U$ be an infinite binary information system. Proposition \ref{P1} allows us
to correspond to the function $h_{U}^{d}(n)$ its type of behavior from the
set $\{\mathrm{LOG},\mathrm{LIN}\}$. Proposition \ref{P2} allows us to
correspond to the function $h_{U}^{a}(n)$ its type of behavior from the set $%
\{\mathrm{CON},\mathrm{LIN}\}$. Propositions \ref{P3} and \ref{P4} allow us to
correspond to each of the functions $L_{U}^{d}(n)$ and $L_{U}^{a}(n)$ its
type of behavior from the set $\{\mathrm{POL},\mathrm{EXP}\}$. A tuple
obtained from the tuple $$%
(h_{U}^{d}(n),h_{U}^{a}(n),L_{U}^{d}(n),L_{U}^{a}(n))$$ by replacing
functions with their types of behavior is called the type of the information
system $U$. We now describe all possible types of infinite binary information systems.

\begin{theorem}
\label{T1}For any infinite binary information system, its type coincides
with one of the rows of Table \ref{tab1}. Each row of  Table \ref{tab1} is the type of some
infinite binary information system.
\end{theorem}

\begin{table}[h]
\centering
\begin{tabular}{|l|llll|} \hline
& $h_{U}^{d}(n)$ & $h_{U}^{a}(n)$ & $L_{U}^{d}(n)$ & $L_{U}^{a}(n)$ \\\hline
1 & $\mathrm{LOG}$ & $\mathrm{CON}$ & $\mathrm{POL}$ & $\mathrm{POL}$ \\
2 & $\mathrm{LIN}$ & $\mathrm{CON}$ & $\mathrm{POL}$ & $\mathrm{POL}$ \\
3 & $\mathrm{LIN}$ & $\mathrm{LIN}$ & $\mathrm{POL}$ & $\mathrm{POL}$ \\
4 & $\mathrm{LIN}$ & $\mathrm{CON}$ & $\mathrm{EXP}$ & $\mathrm{EXP}$ \\
5 & $\mathrm{LIN}$ & $\mathrm{LIN}$ & $\mathrm{EXP}$ & $\mathrm{EXP}$ \\ \hline
\end{tabular}
\caption{Possible types of infinite binary information systems}
  \label{tab1}
\end{table}

For $i=1,\ldots ,5$, we denote by $W_{i}$ the class of all infinite binary
information systems which type coincides with the $i$th row of  Table \ref{tab1}. We now study for each of these complexity classes joint behavior of the depth and number of nodes in decision trees solving problems.

A pair of functions $(\varphi ,\psi )$, where $\varphi :\mathbb{N}%
\rightarrow \mathbb{N}\cup \{0\}$ and $\psi :\mathbb{N}\rightarrow \mathbb{N}%
\cup \{0\}$, is called a boundary $d$-pair of the information system $U$ if,
for any problem $z$ over $U$, there exists a decision tree $\Gamma $ over $U$
which solves the problem $z$ deterministically and for which $h(\Gamma )\leq
\varphi (n)$ and $L(\Gamma )\leq \psi (n)$, where $n=\dim z$. A boundary $d$%
-pair $(\varphi ,\psi )$ of the information system $U$ is called optimal if,
for any boundary $d$-pair $(\varphi ^{\prime },\psi ^{\prime })$ of $U$, the
inequalities $\varphi ^{\prime }(n)\geq \varphi (n)$ and $\psi ^{\prime
}(n)\geq \psi (n)$ hold for any $n\in \mathbb{N}$. An information system $U$
is called $d$-reachable if the pair $(h_{U}^{d},L_{U}^{d})$ is boundary
(and, consequently, optimal boundary) $d$-pair of the system $U$.

A pair of functions $(\varphi ,\psi )$, where $\varphi :\mathbb{N}%
\rightarrow \mathbb{N}\cup \{0\}$ and $\psi :\mathbb{N}\rightarrow \mathbb{N}%
\cup \{0\}$, is called a boundary $a$-pair of the information system $U$ if,
for any problem $z$ over $U$, there exists a decision tree $\Gamma $ over $U$
which solves the problem $z$ nondeterministically and for which $h(\Gamma
)\leq \varphi (n)$ and $L(\Gamma )\leq \psi (n)$, where $n=\dim z$. A
boundary $a$-pair $(\varphi ,\psi )$ of the information system $U$ is called
optimal if, for any boundary $a$-pair $(\varphi ^{\prime },\psi ^{\prime })$
of $U$, the inequalities $\varphi ^{\prime }(n)\geq \varphi (n)$ and $\psi
^{\prime }(n)\geq \psi (n)$ hold for any $n\in \mathbb{N}$. An information
system $U$ is called $a$-reachable if the pair $(h_{U}^{a},L_{U}^{a})$ is
boundary (and, consequently, optimal boundary) $a$-pair of the system $U$.

For deterministic decision trees, the best situation is when the considered information system is $d$-reachable, for nondeterministic decision trees -- when the information system is $a$-reachable.

Each information system from the classes $W_2,W_3,W_4,W_5$ is $d$-reachable. The class $W_1$ contains both information systems that are $d$-reachable and information systems that are not $d$-reachable. For each information system $U$ that is not $d$-reachable, we find a nontrivial boundary $d$-pair which is enough close to the pair $(h_{U}^{d},L_{U}^{d})$.

Each information system from the classes $W_3,W_5$ is $a$-reachable. Each information system from the classes $W_1,W_2,W_4$ is not $a$-reachable. For some information systems $U$ which are not $a$-reachable, we find nontrivial boundary $a$-pairs that are enough close to  $(h_{U}^{a},L_{U}^{a})$. For the rest of information systems $U$ that are not $a$-reachable, the pair $(n,L_{U}^{a}(n))$ is the optimal boundary $a$-pair. Note that, for such information systems, the function $h_{U}^{a}$ is bounded from above by a constant.

The obtained results are related to time-space trade-off for deterministic and nondeterministic decision trees. Details can be found in the following five theorems.

\begin{theorem}
\label{T2} {\rm (a)} The class $W_{1}$ contains both information systems that are $%
d $-reachable and information systems that are not $d$-reachable. Let $U$ be
an information system from $W_{1}$ which is not $d$-reachable. Then, for
each $\varepsilon $, $0<\varepsilon <1$, there exists a positive constant $c$
such that $(c(\log _{2}n)^{1+\varepsilon }+1,2^{c(\log _{2}n)^{1+\varepsilon
}+2})$ is a boundary $d$-pair of the system $U$.

{\rm (b)} Let $U$ be an information system from the class $W_{1}$. Then the system
$U$ is not $a$-reachable and, for each $\varepsilon $, $0<\varepsilon <1$,
there exist positive constants $c_{1}$, $c_{2}$, and $c_{3}$ such that $%
(c_{1},2^{c_{2}(\log _{2}n)^{1+\varepsilon }+c_{3}})$ is a boundary $a$-pair
of the system $U$.
\end{theorem}

\begin{theorem}
\label{T3}Let $U$ be an information system from the class $W_{2}$. Then

{\rm (a)} The system $U$ is $d$-reachable.

{\rm (b)} The system $U$ is not $a$-reachable and $(n,L_{U}^{a}(n))$ is the optimal
boundary $a$-pair of the system $U$.
\end{theorem}

\begin{theorem}
\label{T4}Let $U$ be an information system from the class $W_{3}$. Then

{\rm (a)} The system $U$ is $d$-reachable.

{\rm (b)} The system $U$ is $a$-reachable.
\end{theorem}

\begin{theorem}
\label{T5} The class $W_{4}$ contains both information systems that satisfy the condition of restricted coverage and information systems that do not satisfy this condition.
Let $U$ be an information system from the class $W_{4}$. Then

{\rm (a)} The system $U$ is $d$-reachable.

{\rm (b)} The system $U$ is not $a$-reachable. If the system $U$ does not satisfy
the condition of restricted coverage, then $(n,L_{U}^{a}(n))$ is the optimal
boundary $a$-pair of the system $U$. If the system $U$ satisfies the
condition of restricted coverage, then there exist positive constants $c_{1}$
and $c_{2}$ such that $(c_{1},c\,_{2}^{n})$ is a boundary $a$-pair of the
system $U$.
\end{theorem}

\begin{theorem}
\label{T6}Let $U$ be an information system from the class $W_{5}$. Then

{\rm (a)} The system $U$ is $d$-reachable.

{\rm (b)} The system $U$ is $a$-reachable.
\end{theorem}

Table \ref{tab3} summarizes Theorems \ref{T1}-\ref{T6}. The first column contains name of complexity class. The next four columns describe the type of information systems from this class. The last two columns ``$d$-pairs'' and ``$a$-pairs'' contain information about boundary $d$-pairs and boundary $a$-pairs for information systems from the considered class: ``$d$-reachable'' means that all information systems from the  class are $d$-reachable, ``$a$-reachable'' means that all information systems from the  class are $a$-reachable, Th. \ref{T2} (a), ...,  Th. \ref{T5} (b) are links to corresponding statements Theorem \ref{T2} (a), ..., Theorem  \ref{T5} (b).

\begin{table}[h]
\centering
\begin{tabular}{|l|llll|cc|} \hline
& $h_{U}^{d}(n)$ & $h_{U}^{a}(n)$ & $L_{U}^{d}(n)$ & $L_{U}^{a}(n)$ & $d$-pairs  & $a$-pairs \\\hline
$W_1$ & $\mathrm{LOG}$ & $\mathrm{CON}$ & $\mathrm{POL}$ & $\mathrm{POL}$ & Th. \ref{T2} (a) & Th. \ref{T2} (b) \\
$W_2$ & $\mathrm{LIN}$ & $\mathrm{CON}$ & $\mathrm{POL}$ & $\mathrm{POL}$ & $d$-reachable  & Th. \ref{T3} (b) \\
$W_3$ & $\mathrm{LIN}$ & $\mathrm{LIN}$ & $\mathrm{POL}$ & $\mathrm{POL}$ & $d$-reachable & $a$-reachable \\
$W_4$ & $\mathrm{LIN}$ & $\mathrm{CON}$ & $\mathrm{EXP}$ & $\mathrm{EXP}$ & $d$-reachable & Th. \ref{T5} (b) \\
$W_5$ & $\mathrm{LIN}$ & $\mathrm{LIN}$ & $\mathrm{EXP}$ & $\mathrm{EXP}$ & $d$-reachable & $a$-reachable \\ \hline
\end{tabular}
\caption{Summary of Theorems \ref{T1}-\ref{T6}}
  \label{tab3}
\end{table}

\section{Proofs of  Propositions \ref{P2}-\ref{P4}} \label{S3}

In this section, we prove a number of auxiliary statements and the three mentioned propositions.

\begin{lemma}
\label{L0}Let $U=(A,F)$ be an infinite binary information system. Then

{\rm (a)} If $U$ satisfies the condition of coverage with parameter $m$, then for
any $n\in \mathbb{N}$, any $(n,U)$-set is a union of a finite number of $%
(m,U)$-sets.

{\rm (b)} If $U$ satisfies the condition of restricted coverage with parameters $m$
and $t$, then for any $n\in \mathbb{N}$, any $(n,U)$-set is a union of at
most $t^{n}$ $(m,U)$-sets.
\end{lemma}

\begin{proof}
(a) Let $U$ satisfy the condition of coverage with parameter $m$: any $%
(m+1,U)$-set is a union of a finite number of $(m,U)$-sets. We now show by
induction on $n$ that, for any $n\in \mathbb{N}$, any $(n,U)$-set is a union
of a finite number of $(m,U)$-sets. If $n\leq m+1$, then evidently, the
considered statement holds. Let for some $n$, $n\geq m+1$, the considered
statement hold. Let us show that it holds for $n+1$. We consider an
arbitrary $(n+1,U)$-set $B$ which is  the set
of solutions from $A$ of a $(n+1,U)$-system of equations%
\[
S=\{f_{1}(x)=\delta _{1},\ldots ,f_{n+1}(x)=\delta _{n+1}\}.
\]%
Let us consider a $(n,U)$-system of equations
\[
S^{\prime }=\{f_{1}(x)=\delta _{1},\ldots ,f_{n}(x)=\delta _{n}\}
\]%
and the set $B^{\prime }$ of solutions from $A$ of this system. Then $%
B^{\prime }$ is a $(n,U)$-set and, according to the inductive hypothesis, $%
B^{\prime }$ is a union of a finite number of $(m,U)$-sets $B_{1},\ldots
,B_{k}$, where for $i=1,\ldots ,k$, the set $B_{i}$ is the set of solutions
of an $(m,U)$-system of equations $S_{i}$. One can
show that the set $B$ is equal to the union of the sets of solutions on $A$
of the systems of equations $S_{i}\cup $ $\{f_{n+1}(x)=\delta _{n+1}\}$, $%
i=1,\ldots ,k$. Each of these sets is an $(m+1,U)$-set and therefore is a union of
a finite number of $(m,U)$-sets. Thus, $B$ is a union of a finite number
of $(m,U)$-sets.

(b) Let $U$ satisfy the condition of restricted coverage with parameters $%
m $ and $t$: any $(m+1,U)$-set is a union of at most $t$ $(m,U)$-sets. We
now show by induction on $n$ that, for any $n\in \mathbb{N}$, any $(n,U)$%
-set is a union of at most $t^{n}$ $(m,U)$-sets. If $n\leq m+1$, then
evidently, the considered statement holds. Let for some $n$, $n\geq m+1$,
the considered statement hold. Let us show that it holds for $n+1$. We
consider an arbitrary $(n+1,U)$-set $B$ which is the set of solutions from $A$ of a $(n+1,U)$-system of equations%
\[
S=\{f_{1}(x)=\delta _{1},\ldots ,f_{n+1}(x)=\delta _{n+1}\}.
\]%
Let us consider a $(n,U)$-system of equations
\[
S^{\prime }=\{f_{1}(x)=\delta _{1},\ldots ,f_{n}(x)=\delta _{n}\}
\]%
and the set $B^{\prime }$ of solutions from $A$ of this system. Then $%
B^{\prime }$ is a $(n,U)$-set and, according to the inductive hypothesis, $%
B^{\prime }$ is a union of $(m,U)$-sets $B_{1},\ldots ,B_{k}$, where $k\leq
t^{n}$ and, for $i=1,\ldots ,k$, the set $B_{i}$ is the set of solutions from $A$ of
an $(m,U)$-system of equations $S_{i}$. One can show
that the set $B$ is equal to the union of the sets of solutions on $A$ of
the systems of equations $S_{i}\cup $ $\{f_{n+1}(x)=\delta _{n+1}\}$, $%
i=1,\ldots ,k$. Each of these sets is an $(m+1,U)$-set and therefore is a union of
at most $t$ $(m,U)$-sets. Thus, $B$ is a union of at most $t^{n+1}$ $%
(m,U)$-sets.
\end{proof}

\begin{proof}
[Proof of Proposition \ref{P2}]
(\textrm{CON}) Let $U$ satisfy the condition of
coverage: there exists $m\in \mathbb{N}$ such that any $(m+1,U)$-set is a
union of a finite number of $(m,U)$-sets. From  Lemma \ref{L0} it follows
that, for any $n\in \mathbb{N}$, any $(n,U)$-set is a union of a finite
number of $(m,U)$-sets.

Let $z=(\nu ,f_{1},\ldots ,f_{n})$ be a problem over $U$. We now show
that $h_{U}^{a}(z)\leq m$. Let $\bar{\delta}=(\delta _{1},\ldots ,\delta
_{n})$ be a tuple from $\{0,1\}^{n}$ such that the equation system
\[
S(\bar{\delta})=\{f_{1}(x)=\delta _{1},\ldots ,f_{n}(x)=\delta _{n}\}
\]%
has a solution from $A$. For each solution $a\in A$ of this system, we have $%
z(a)=\nu (\bar{\delta})$. The set of solutions of $S(\bar{\delta})$ is a
union of a finite number of $(m,U)$-sets $D_{1},\ldots ,D_{s}$. Each of
these sets $D_{i}$ is the set of solutions of an $(m,U)$-system of
equations. For the considered $(m,U)$-system of
equations, we construct a complete path $\xi _{i}$ with $m$ working
nodes such that $A(\xi _{i})=D_i$ and the terminal node of $\xi _{i}$ is
labeled with the number $\nu (\bar{\delta})$. Denote $\Sigma (\bar{\delta}%
)=\{\xi _{1},\ldots ,\xi _{s}\}$. Let $\Sigma =\bigcup \Sigma (\bar{\delta}%
) $, where the union is considered among all tuples $\bar{\delta}\in
\{0,1\}^{n} $ such that the system of equations $S(\bar{\delta})$ has a
solution from $A$. We identify initial nodes of all paths from $\Sigma $. As
a result, we obtain a decision tree over the information system $U$ which
solves the problem $z$ nondeterministically and which depth is at most $m$.
Taking into account that $z$ is an arbitrary problem over $U$, we obtain $%
h_{U}^{a}(n)\leq m$ for any $n\in \mathbb{N}$.

(\textrm{LIN}) Let $U$ do not satisfy the condition of coverage. Assume
that there exists $m\in \mathbb{N}$ such that $h_{U}^{a}(m+1)\leq m$.
Let $B$ be an arbitrary $(m+1,U)$-set given by an $(m+1,U)$-system of equations $S$.
We show that  the set $B$  is a union of a finite
number of $(m,U)$-sets. Let
\[
S=\{f_{1}(x)=\delta _{1},\ldots ,f_{m+1}(x)=\delta _{m+1}\}\text{.}
\]

Consider the problem $z=(\nu ,f_{1},\ldots ,f_{m+1})$ over $U$ such that,
for any $\bar{\sigma}\in \{0,1\}^{m+1}$, $\nu (\bar{\sigma})\in \{1,2\}$ and
$\nu (\bar{\sigma})=1$ if and only if $\bar{\sigma}=(\delta _{1},\ldots
,\delta _{m+1})$.

Let $\Gamma $ be a decision tree which solves the problem $z$
nondeterministically and for which $h(\Gamma )\leq m$. Choose all complete
paths $\xi $ in $\Gamma $ in which the terminal node is labeled with the
number $1$ and the set $A(\xi )$ is nonempty. Each such path $\xi $
describes a $(m,U)$-set $A(\xi )$. The number of these paths is finite. The
union of sets described by these paths is equal to $B$. Since $B$ is an arbitrary $(m+1,U)$-set, the
information system $U$ satisfies the condition of coverage with parameter $m$, but this is
impossible. Thus, $h_{U}^{a}(n)=n$ for any $n\in \mathbb{N}$.
\end{proof}

We now consider some auxiliary statements related to the space complexity of
decision trees. Let $\Gamma $ be a decision tree over an information system $%
U=(A,F)$ and $d$ be an edge of $\Gamma $ entering a node $w$. We denote by $%
\Gamma (d)$ a subtree of $\Gamma $ which root is the node $w$. We say that a
complete path $\xi $ of $\Gamma $ is realizable if $A(\xi )\neq \emptyset $.

\begin{lemma}
\label{L1}Let $U=(A,F)$ be an infinite binary information system, $z=(\nu
,f_{1},\ldots ,f_{n})$ be a problem over $U$, and $\Gamma $ be a decision
tree over $U$ which solves the problem $z$ deterministically and for which $%
L(\Gamma )=L_{U}^{d}(z)$. Then

{\rm (a)} Each working node of $\Gamma $ has two edges leaving this node.

{\rm (b)} For each node of $\Gamma $, there exists a realizable complete path that
passes through this node.
\end{lemma}

\begin{proof}
(a) It is clear that there exists at least one realizable complete path that
passes through the root of $\Gamma $. Let us assume that $w$ be a node of $\Gamma $
different from the root and such that there is no a realizable complete path
which passes through $w$. Let $d$ be an edge entering the node $w$. We
remove from $\Gamma $ the edge $d$ and the subtree $\Gamma (d)$. As a result,
we obtain a decision tree $\Gamma ^{\prime }$ which solves $z$
deterministically and for which $L(\Gamma ^{\prime })<L(\Gamma )$ but this
is impossible.

(b) Let us assume that in $\Gamma $ there exists a working node $w$ which
has only one leaving edge $d$ entering a node $w_{1}$. We remove from $%
\Gamma $ the node $w$ and the edge $d$ and connect the edge entering the node $w$
to the node $w_{1}$. As a result, we obtain a decision tree $\Gamma ^{\prime
} $ which solves the problem $z$ deterministically and for which $L(\Gamma
^{\prime })<L(\Gamma )$ but this is impossible.
\end{proof}

Let $U$ be an infinite binary information system, $\Gamma $ be a decision
tree over $U$, and $d$ be an edge of $\Gamma $. The subtree $\Gamma (d)$ is
called full if there exist edges $d_{1},\ldots ,d_{m}$ in $\Gamma (d)$ such
that the removal of these edges and subtrees $\Gamma (d_{1}),\ldots ,\Gamma
(d_{m})$ transforms the subtree $\Gamma (d)$ into a tree $G$ such that
each terminal node of $G$ is a terminal node of $\Gamma $, and exactly two
edges labeled with the numbers $0$ and $1$ respectively leave each working node
of $G$.

\begin{lemma}
\label{L2}Let $U=(A,F)$ be an infinite binary information system, $z=(\nu
,f_{1},\ldots ,f_{n})$ be a problem over $U$, and $\Gamma $ be a decision
tree over $U$ which solves the problem $z$ nondeterministically and for
which $L(\Gamma )=L_{U}^{a}(z)$. Then

{\rm (a)} For each node of $\Gamma $, there exists a realizable complete path that
passes through this node.

{\rm (b)} If a working node $w$ of $\Gamma $ has $m$ leaving edges $d_{1},\ldots
,d_{m}$ labeled with the same number and $m\geq 2$, then the subtrees $%
\Gamma (d_{1}),\ldots ,\Gamma (d_{m})$ are not full.

{\rm (c)} If the root $r$ of $\Gamma $ has $m$ leaving edges $d_{1},\ldots ,d_{m}$
and $m\geq 2$, then the subtrees $\Gamma (d_{1}),\ldots ,\Gamma (d_{m})$ are
not full.
\end{lemma}

\begin{proof} (a) It is clear that there exists at least one realizable complete path that passes through the root of $\Gamma $.
Let us assume that $w$ be a node of $\Gamma $ different from the root
and such that there is no a realizable complete path which passes through $w$%
. Let $d$ be an edge entering the node $w$. We remove from $\Gamma $ the
edge $d$ and the subtree $\Gamma (d)$. As a result, we obtain a decision tree
$\Gamma ^{\prime }$ which solves the problem $z$ nondeterministically and for which $%
L(\Gamma ^{\prime })<L(\Gamma )$ but this is impossible.

(b) Let $w$ be a working node of $\Gamma $ which has $m$ leaving edges $%
d_{1},\ldots ,d_{m}$ labeled with the same number, $m\geq 2$, and at least
one of the subtrees $\Gamma (d_{1}),\ldots ,\Gamma (d_{m})$ is full. For the
definiteness, we assume that $\Gamma (d_{1})$ is full. Remove from $\Gamma $
the edges $d_{2},\ldots ,d_{m}$ and subtrees $\Gamma (d_{2}),\ldots ,\Gamma
(d_{m})$. We now show that the obtained tree $\Gamma ^{\prime }$ solves the
problem $z$ nondeterministically. Assume the contrary. Then there exists an
object $a\in A$ such that, for each complete path $\xi $ with $a\in A(\xi )$%
, the path $\xi $ passes through one of the edges $d_{2},\ldots ,d_{m}$ but
it is not true. Let $\xi $ be a complete path such that $a\in A(\xi )$. Then, according to the assumption,
this path passes through the node $w$. Let $\xi ^{\prime }$ be the part of
this path from the root of $\Gamma $ to the node $w$. Since the edges $d_{1},\ldots ,d_{m}$ are labeled with the same number and $\Gamma (d_{1})$
is a full subtree, we can find in $\Gamma (d_1)$ the continuation of $\xi
^{\prime }$ to a terminal node of $\Gamma (d_{1})$ such that the obtained
complete path $\xi ^{\prime \prime }$ of $\Gamma $ satisfies the condition $%
a\in A(\xi ^{\prime \prime })$. Hence $\Gamma ^{\prime }$ is a decision tree
which solves the problem $z$ nondeterministically and for which $L(\Gamma
^{\prime })<L(\Gamma )$ but this is impossible.

(c) The part (c) of the statement can be proven in the same way as the part
(b).
\end{proof}

We now prove a number of statements about classes of decision trees. Let $%
\Gamma $ be a decision tree. We denote by $L_{t}(\Gamma )$ the number of
terminal nodes in $\Gamma $ and by $L_{w}(\Gamma )$ -- the number of working
nodes in $\Gamma $. It is clear that $L(\Gamma )=1+L_{t}(\Gamma
)+L_{w}(\Gamma )$.

Let $U$ be an infinite binary information systems. We denote by $G_{d}(U)$
the set of all deterministic decision trees over $U$ and by $G_{d}^{2}(U)$
-- the set of all decision trees from $G_{d}(U)$ such that each working node
of the tree has two leaving edges.

\begin{lemma}
\label{L3}Let $U$ be an infinite binary information system. Then

{\rm (a)} If $\Gamma \in $ $G_{d}^{2}(U)$, then $L_{w}(\Gamma )=L_{t}(\Gamma )-1$.

{\rm (b)} If $\Gamma \in $ $G_{d}(U)\setminus G_{d}^{2}(U)$, then $L_{w}(\Gamma
)>L_{t}(\Gamma )-1$.
\end{lemma}

\begin{proof}
(a) We prove the equality $L_{w}(\Gamma )=L_{t}(\Gamma )-1$ for trees from $%
G_{d}^{2}(U)$ by induction on $L_{t}(\Gamma )$. If $L_{t}(\Gamma )\leq 2$,
then this equality holds. Let $m\geq 2$ and, for each $\Gamma \in
G_{d}^{2}(U)$ with $L_{t}(\Gamma )\leq m$, the considered equality hold.
Let $\Gamma \in G_{d}^{2}(U)$ and $L_{t}(\Gamma )=m+1$. It is clear that
there exists a node $w$ of the tree $\Gamma $ such that all children of $w$
are terminal nodes. Remove children of $w$ and edges entering these children
and attach to $w$ a number from $\mathbb{N}$. We denote the obtained tree by
$\Gamma ^{\prime }$. It is clear that $\Gamma ^{\prime }\in G_{d}^{2}(U)$
and $L_{t}(\Gamma ^{\prime })=m$. By the inductive hypothesis, $L_{w}(\Gamma
^{\prime })=L_{t}(\Gamma ^{\prime })-1$. Taking into account that $%
L_{w}(\Gamma ^{\prime })=L_{w}(\Gamma )-1$ and $L_{t}(\Gamma ^{\prime
})=L_{t}(\Gamma )-1$, we obtain $L_{w}(\Gamma )=L_{t}(\Gamma )-1$.

(b) Let $\Gamma \in $ $G_{d}(U)\setminus G_{d}^{2}(U)$ and there be $m\geq
1 $ working nodes in $\Gamma $ each of which has exactly one leaving edge.
We add $m$ new terminal nodes to $\Gamma $ and, as a result, obtain a tree $%
\Gamma ^{\prime }\in $ $G_{d}^{2}(U)$. Then $L_{w}(\Gamma ^{\prime
})=L_{t}(\Gamma ^{\prime })-1$, $L_{w}(\Gamma ^{\prime })=L_{w}(\Gamma )$,
and $L_{t}(\Gamma ^{\prime })=L_{t}(\Gamma )+m$. Therefore $L_{w}(\Gamma
)=L_{t}(\Gamma )+m-1$. Since $m\geq 1$, we obtain $L_{w}(\Gamma
)>L_{t}(\Gamma )-1$.
\end{proof}

We denote by $G_{a}^{f}(U)$ the set of all decision trees $\Gamma $ over $U$
that satisfy the following conditions: (i) if a working node of $\Gamma $
has $m$ leaving edges $d_{1},\ldots ,d_{m}$ labeled with the same number and
$m\geq 2$, then the subtrees $\Gamma (d_{1}),\ldots ,\Gamma (d_{m})$ are not
full, and (ii) if the root of $\Gamma $ has $m$ leaving edges $d_{1},\ldots
,d_{m}$ and $m\geq 2$, then the subtrees $\Gamma (d_{1}),\ldots ,\Gamma
(d_{m})$ are not full. One can show that $G_{d}^{2}(U)\subseteq
G_{d}(U)\subseteq G_{a}^{f}(U)$.

\begin{lemma}
\label{L4}Let $U$ be an infinite binary information system. If $\Gamma \in $
$G_{a}^{f}(U)\setminus G_{d}^{2}(U)$, then $L_{w}(\Gamma )>L_{t}(\Gamma )-1$.
\end{lemma}

\begin{proof}
We prove the considered statement by induction on $L_{t}(\Gamma )$. Let $%
\Gamma \in $ $G_{a}^{f}(U)\setminus G_{d}^{2}(U)$ and $L_{t}(\Gamma )=1$.
Then $\Gamma $ consists of one complete path with $k\geq 0$ working nodes.
If $k=0$ then $\Gamma \in G_{d}^{2}(U)$ but this is impossible. Therefore $%
L_{w}(\Gamma )=k\geq 1$ and $L_{t}(\Gamma )=1$. Hence $L_{w}(\Gamma
)>L_{t}(\Gamma )-1$.

Let, for some $m\geq 1$ for each  decision tree $\Gamma \in G_{a}^{f}(U)\setminus G_{d}^{2}(U)$
with $L_{t}(\Gamma )\leq m$, the inequality $L_{w}(\Gamma )>L_{t}(\Gamma )-1$
hold. Consider a decision tree $\Gamma $ such that $\Gamma \in $ $%
G_{a}^{f}(U)\setminus G_{d}^{2}(U)$ and $L_{t}(\Gamma )=m+1$. We now show
that $L_{w}(\Gamma )>L_{t}(\Gamma )-1$. If $\Gamma \in G_{d}(U)\setminus
G_{d}^{2}(U)$, then by  Lemma \ref{L3}, $L_{w}(\Gamma )>L_{t}(\Gamma )-1$.
Let $\Gamma \in $ $G_{a}^{f}(U)\setminus G_{d}(U)$. Then the tree $\Gamma $
contains a node $v$ (the root or a working node) which has two leaving edges
labeled with the same number (if $v$ is a working node) or do not labeled
with numbers (if $v$ is the root). We call such edges equally labeled. Since
$\Gamma $ is a finite tree, there is a node $w$ of $\Gamma $ which has two
leaving edges $d_{1}$ and $d_{2}$ that are equally labeled, and in the
subtrees $\Gamma (d_{1})$ and $\Gamma (d_{2})$ there are no nodes with two
leaving edges that are equally labeled.

It is clear that the subtree $\Gamma (d_{1})$ is not full. We add to this
tree a node and an edge leaving this node and entering the root of $\Gamma
(d_{1})$. As a result, we obtain a decision tree from $G_{d}(U)\setminus
G_{d}^{2}(U)$. Using  Lemma \ref{L3} we obtain $L_{w}(\Gamma
(d_{1}))>L_{t}(\Gamma (d_{1}))-1$.

Remove from $\Gamma $ the edge $d_{1}$ and the subtree $\Gamma (d_{1})$.
Denote by $\Gamma ^{\prime }$ the obtained tree. Since the subtree $\Gamma
(d_{2})$ is not full, $\Gamma ^{\prime }\notin G_{d}^{2}(U)$. It is clear
that $\Gamma ^{\prime }\in G_{a}^{f}(U)$. One can show that $L_{t}(\Gamma
^{\prime })<L_{t}(\Gamma )$. By the inductive hypothesis, $L_{w}(\Gamma ^{\prime
})>L_{t}(\Gamma ^{\prime })-1$, and hence $L_{w}(\Gamma ^{\prime })\geq
L_{t}(\Gamma ^{\prime })$. Therefore $L_{w}(\Gamma ^{\prime })+L_{w}(\Gamma
(d_{1}))>L_{t}(\Gamma ^{\prime })+L_{t}(\Gamma (d_{1}))-1$. Since $%
L_{w}(\Gamma )=L_{w}(\Gamma ^{\prime })+L_{w}(\Gamma (d_{1}))$ and $%
L_{t}(\Gamma )=L_{t}(\Gamma ^{\prime })+L_{t}(\Gamma (d_{1}))$, we obtain $%
L_{w}(\Gamma )>L_{t}(\Gamma )-1$.
\end{proof}

Let $U=(A,F)$ be an infinite binary information system. For $f_{1},\ldots
,f_{n}\in F$ we denote by $N_{U}(f_{1},\ldots ,f_{n})$ the number of $n$%
-tuples $(\delta _{1},\ldots ,\delta _{n})\in \{0,1\}^{n}$ for which the
system of equations
\[
\{f_{1}(x)=\delta _{1},\ldots ,f_{n}(x)=\delta _{n}\}
\]%
has a solution from $A$. For $n\in \mathbb{N}$, denote
\[
N_{U}(n)=\max \{N_{U}(f_{1},\ldots ,f_{n}):f_{1},\ldots ,f_{n}\in F\}.
\]

It is clear that, for any $m,n\in \mathbb{N}$, if $m\leq n$ then $%
N_{U}(m)\leq N_{U}(n)$.

\begin{proposition}
\label{P5} Let $U=(A,F)$ be an infinite binary information system. Then, for
any $n\in \mathbb{N}$,
\[
L_{U}^{a}(n)=L_{U}^{d}(n)=2N_{U}(n).
\]
\end{proposition}

\begin{proof}
Let $z=(\nu ,f_{1},\ldots ,f_{m})$ be a problem over $U$ and $m\leq n$. Let $%
\Gamma $ be a decision tree over $U$ which solves the problem $z$
deterministically, uses only attributes from the set $\{f_{1},\ldots
,f_{m}\} $, and has minimum number of nodes among such decision trees. In
the same way as in the proof of  Lemma \ref{L1}, one can prove that each
working node of $\Gamma $ has two edges leaving this node and, for each node
of $\Gamma $, there exists a realizable complete path that passes through
this node. Let $\xi _{1}$ and $\xi _{2}$ be different complete paths in $%
\Gamma $, $a_{1}\in A(\xi _{1})$, and $a_{2}\in A(\xi _{2})$. It is easy to
show that $(f_{1}(a_{1}),\ldots ,f_{m}(a_{1}))\neq (f_{1}(a_{2}),\ldots
,f_{m}(a_{2}))$. Therefore $L_{t}(\Gamma )\leq N_{U}(f_{1},\ldots
,f_{m})\leq N_{U}(n)$.
It is clear that $\Gamma  \in G_d^2(U)$.
By  Lemma \ref{L3}, $L_{w}(\Gamma )=L_{t}(\Gamma )-1$.
Hence $L(\Gamma )\leq 2N_{U}(n)$. Taking into account that $z$ is an
arbitrary problem over $U$ with $\dim z\leq n$ we obtain
\[
L_{U}^{d}(n)\leq 2N_{U}(n).
\]

Since any decision tree solving the problem $z$ deterministically solves
it nondeterministically we obtain%
\[
L_{U}^{a}(n)\leq L_{U}^{d}(n).
\]%
We now show that $2N_{U}(n)\leq L_{U}^{a}(n)$. Let us consider
a problem $z=(\nu ,f_{1},\ldots ,f_{n})$ over $U$ such that  $$%
N_{U}(f_{1},\ldots ,f_{n})=N_{U}(n)$$ and, for any $\bar{\delta}_{1},\bar{%
\delta}_{2}\in \{0,1\}^{n}$, if $\bar{\delta}_{1}\neq \bar{\delta}_{2}$,
then $\nu (\bar{\delta}_{1})\neq \nu (\bar{\delta}_{2})$. Let $\Gamma $ be a
decision tree over $U$ which solves the problem $z$ nondeterministically and
for which $L(\Gamma )=L_{U}^{a}(z)$. By  Lemma \ref{L2}, $\Gamma \in
G_{a}^{f}(U)$. Using Lemmas \ref{L3} and \ref{L4} we obtain $L_{w}(\Gamma
)\ge L_{t}(\Gamma )-1$. It is clear that $L_{t}(\Gamma )\geq N_{U}(f_{1},\ldots
,f_{n})=N_{U}(n)$. Therefore $L(\Gamma )\geq 2N_{U}(n)$, $L_{U}^{a}(z)\geq
2N_{U}(n)$, and $L_{U}^{a}(n)\geq 2N_{U}(n)$.
\end{proof}

The next statement follows directly from Lemmas 5.1 and 5.2 \cite{Moshkov05} and
evident inequality $N_{U}(n)\leq 2^{n}$ which is true for any infinite
binary information system $U$. The proof of Lemma 5.1 from \cite{Moshkov05} is based on Theorems 4.6 and 4.7 from the same monograph that are similar to results obtained in \cite{Sauer72,Shelah72}.

\begin{proposition}
\label{P6} For any infinite binary information system $U$, the function $%
N_{U}(n)$ has one of the following two types of behavior:

{\rm (POL)} If the system $U$ has finite I-dimension, then for any $%
n\in \mathbb{N}$,
\[
n+1\leq N_{U}(n)\leq (4n)^{I(U)}.
\]

{\rm(EXP)} If the system $U$ has infinite I-dimension, then for any $%
n\in \mathbb{N}$,
\[
N_{U}(n)=2^{n}.
\]
\end{proposition}

We now prove Propositions  \ref{P3} and \ref{P4}.

\begin{proof}
[Proof of Proposition \ref{P3}] The statement of the proposition follows immediately
from Propositions \ref{P5} and \ref{P6}.
\end{proof}

\begin{proof}
[Proof of Proposition \ref{P4}] The statement of the proposition follows immediately
from  Proposition \ref{P5}.
\end{proof}

\section{Proof of  Theorem \protect\ref{T1}} \label{S4}

First, we prove six auxiliary statements.

\begin{lemma}
\label{L5}For any infinite binary information system, its type coincides
with one of the rows of  Table \ref{tab1}.
\end{lemma}

\begin{proof}
To prove this statement we fill  Table \ref{tab2}. In the first column
\textquotedblleft Cover.\textquotedblright , we have either
\textquotedblleft Yes\textquotedblright\ or \textquotedblleft
No\textquotedblright : \textquotedblleft Yes\textquotedblright\ if the
considered information system satisfies the condition of coverage and
\textquotedblleft No\textquotedblright\ otherwise. In the second column
\textquotedblleft Restr. cover.\textquotedblright ,\ we also have
either \textquotedblleft Yes\textquotedblright\ or \textquotedblleft
No\textquotedblright : \textquotedblleft Yes\textquotedblright\ if the
considered information system satisfies the condition of restricted coverage
and \textquotedblleft No\textquotedblright\ otherwise. In the third column
\textquotedblleft I-dim.\textquotedblright\ we have either
\textquotedblleft Fin\textquotedblright\ or \textquotedblleft
Inf\textquotedblright : \textquotedblleft Fin\textquotedblright\ if the
considered information system has finite I-dimension and \textquotedblleft
Inf\textquotedblright\ if the considered information system has infinite $I$%
-dimension.

If an information system does not satisfy the condition of
coverage, then this information system does not satisfy the condition of
restricted coverage. It means that there are only six possible tuples of
values of the considered three parameters of information systems which
correspond to the six rows of Table \ref{tab2}. The values of the considered three
parameters define the types of behavior of functions $h_{U}^{d}(n)$, $%
h_{U}^{a}(n)$, $L_{U}^{d}(n)$, and $L_{U}^{a}(n)$ according to Propositions %
\ref{P1}-\ref{P4}. We see that the set of possible tuples of values in the
last four columns coincides with the set of rows of  Table \ref{tab1}.
\end{proof}

\begin{table}[h]
\centering
\begin{tabular}{|lllllll|}\hline
Cover. & Restr. & I-dim. & $h_{U}^{d}(n)$ & $h_{U}^{a}(n)$ & $%
L_{U}^{d}(n)$ & $L_{U}^{a}(n)$ \\
& cover. &  &  &  &  &  \\ \hline
Yes & Yes & Fin & $\mathrm{LOG}$ & $\mathrm{CON}$ & $\mathrm{POL}$ & $%
\mathrm{POL}$ \\
Yes & No & Fin & $\mathrm{LIN}$ & $\mathrm{CON}$ & $\mathrm{POL}$ & $\mathrm{%
POL}$ \\
No & No & Fin & $\mathrm{LIN}$ & $\mathrm{LIN}$ & $\mathrm{POL}$ & $\mathrm{%
POL}$ \\
Yes & Yes & Inf & $\mathrm{LIN}$ & $\mathrm{CON}$ & $\mathrm{EXP}$ & $%
\mathrm{EXP}$ \\
Yes & No & Inf & $\mathrm{LIN}$ & $\mathrm{CON}$ & $\mathrm{EXP}$ & $\mathrm{%
EXP}$ \\
No & No & Inf & $\mathrm{LIN}$ & $\mathrm{LIN}$ & $\mathrm{EXP}$ & $\mathrm{EXP}$ \\ \hline
\end{tabular}
\caption{Parameters and types of infinite binary information systems}
  \label{tab2}
\end{table}

For each row of Table \ref{tab1}, we consider an example of infinite binary
information system which type coincides with this row.

For any $i\in \mathbb{N}$, we define two functions $p_{i}:\mathbb{N}%
\rightarrow \{0,1\}$ and $l_{i}:\mathbb{N}\rightarrow \{0,1\}$. Let $j\in
\mathbb{N}$. Then $p_{i}(j)=1$ if and only if $j=i$, and $l_{i}(j)=1$ if and
only if $j>i$.

Define an information system $U_{1}=(A_{1},F_{1})$ as follows: $A_{1}=%
\mathbb{N}$ and $F_{1}=\{l_{i}:i\in \mathbb{N}\}$.

\begin{lemma}
\label{L6.1} The information system $U_{1}$ belongs to the class $W_{1}$, $%
h_{U_{1}}^{d}(n)=\lceil \log _{2}(n+1)\rceil $, $h_{U_{1}}^{a}(1)=1$ and $%
h_{U_{1}}^{a}(n)=2$ if $n>1$, $L_{U_{1}}^{d}(n)=2(n+1)$, and $%
L_{U_{1}}^{a}(n)=2(n+1)$ for any $n\in \mathbb{N}$. This information system
satisfies the condition of coverage with  parameter $3$, satisfies the
condition of restricted coverage with parameters $3$ and $1$, and has
finite I-dimension equals to $1$. The information system $U_{1}$ is $d$%
-reachable.
\end{lemma}

\begin{proof}
It is easy to show that $N_{U_{1}}(n)=n+1$ for any $n\in \mathbb{N}$. Using
 Proposition \ref{P5} we obtain $L_{U_{1}}^{d}(n)=L_{U_{1}}^{a}(n)=2(n+1)$
for any $n\in \mathbb{N}$. Let $n\in \mathbb{N}$. Consider a problem $z=(\nu
,l_{1},\ldots ,l_{n})$ over $U_{1}$ such that, for each $\bar{\delta}_{1},%
\bar{\delta}_{2}\in \{0,1\}^{n}$ with $\bar{\delta}_{1}\neq \bar{\delta}_{2}$%
, $\nu (\bar{\delta}_{1})\neq \nu (\bar{\delta}_{2})$. It is clear that $%
N_{U_{1}}(l_{1},\ldots ,l_{n})=n+1$. Therefore each decision tree $\Gamma $
over $U_{1}$ that solves the problem $z$ deterministically has at least $n+1$
terminal nodes. One can show that the number of terminal nodes in $\Gamma $
is at most $2^{h(\Gamma )}$. Hence $n+1\leq 2^{h(\Gamma )}$ and $\log
_{2}(n+1)\leq h(\Gamma )$. Since $h(\Gamma )$ is an integer, $\lceil \log
_{2}(n+1)\rceil \leq h(\Gamma )$. Thus, $h_{U_{1}}^{d}(n)\geq \lceil \log
_{2}(n+1)\rceil $. Set $m=\lceil \log _{2}(n+1)\rceil $. Then $n\leq 2^{m}-1$%
. One can show that $h_{U_{1}}^{d}(2^{m}-1)\leq m$ (the construction of an
appropriate decision tree is based on an analog of binary search, and we use
only attributes from the problem description) and $h_{U_{1}}^{d}(n)\leq
h_{U_{1}}^{d}(2^{m}-1)$. Therefore $h_{U_{1}}^{d}(n)\leq \lceil \log
_{2}(n+1)\rceil $ and $h_{U_{1}}^{d}(n)=\lceil \log _{2}(n+1)\rceil $. It is
clear that $h_{U_{1}}^{a}(1)=1$. Let $n\geq 2$ and $z=(\nu ,f_{1},\ldots
,f_{n})$ be an arbitrary problem over $U_{1}$ and $l_{i_{1}},\ldots
,l_{i_{m}}$ be all pairwise different attributes from the set $%
\{f_{1},\ldots ,f_{n}\}$ ordered such that $i_{1}<\ldots <i_{m}$. Then these
attributes divide the set $\mathbb{N}$ into $m+1$ nonempty domains that are
sets of solutions on $\mathbb{N}$ of the following systems of equations: $%
\{l_{i_{1}}(x)=0\}$, $\{l_{i_{1}}(x)=1,l_{i_{2}}(x)=0\}$, \ldots , $%
\{l_{i_{m-1}}(x)=1,l_{i_{m}}(x)=0\}$, $\{l_{i_{m}}(x)=1\}$. The value $z(x)$
is constant in each of the considered domains. Using these facts it is easy
to show that there exists a decision tree $\Gamma $ over $U_{1}$ which
solves the problem $z$ nondeterministically and for which $h(\Gamma )=2$ if $%
m\geq 2$. Therefore $h_{U_{1}}^{a}(n)\leq 2$. One can show that there
exists a problem $z$ over $U_{1}$ such that $\dim z=n$ and $%
h_{U_{1}}^{a}(z)\geq 2$. Therefore $h_{U_{1}}^{a}(n)=2$.

Since the function $%
h_{U_{1}}^{d}$ has the type of behavior \textrm{LOG}, the information system
$U_{1}$ belongs to the class $W_{1}$ -- see Table  \ref{tab1}. One can show that the
information system $U_{1}$ satisfies the condition of coverage with
parameter $3$, satisfies the condition of restricted coverage with
parameters $3$ and $1$, and has finite I-dimension equals to $1$.

Let $%
z=(\nu ,f_{1},\ldots ,f_{n})$ be a problem over $U_{1}$. We know
that there exists a decision tree $\Gamma $ over $U_{1}$ which solves this
problem deterministically, uses only attributes from the set $\{f_{1},\ldots
,f_{n}\}$, and which depth is at most $\lceil \log _{2}(n+1)\rceil $. By
removal of some nodes and edges from $\Gamma $ we can obtain a decision tree
$\Gamma ^{\prime }$ over $U_{1}$ which solves the problem $z$
deterministically, and in which each working node has exactly two leaving
edges and each complete path is realizable. It is clear that $%
N_{U_{1}}(f_{1},\ldots ,f_{n})\leq n+1$. Therefore $L_{t}(\Gamma ^{\prime
})\leq n+1$. By  Lemma \ref{L3}, $L_{w}(\Gamma ^{\prime })\leq n$. Therefore $%
L(\Gamma ^{\prime })\leq 2(n+1)=L_{U_{1}}^{d}(n)$. Taking into account that $%
h(\Gamma ^{\prime })\leq \lceil \log _{2}(n+1)\rceil =h_{U_{1}}^{d}(n)$ and $z$ is an
arbitrary problem over $U_{1}$ with $\dim z=n$ we obtain that $U_{1}$ is $d$%
-reachable.
\end{proof}

Define an information system $U_{2}=(A_{2},F_{2})$ as follows: $A_{2}=%
\mathbb{N}$ and $F_{2}=\{p_{i}:i\in \mathbb{N}\}\cup \{l_{2^{i}}:i\in
\mathbb{N}\}$.

\begin{lemma}
\label{L6.2}The information system $U_{2}$ belongs to the class $W_{2}$, $%
h_{U_{2}}^{d}(n)=n$, $h_{U_{2}}^{a}(n)=1$, $L_{U_{2}}^{d}(n)=2(n+1)$, and $%
L_{U_{2}}^{a}(n)=2(n+1)$ for any $n\in \mathbb{N}$. This information system
satisfies the condition of coverage with  parameter $2$, does not satisfy
the condition of the restricted coverage, and has finite I-dimension
equals to $1$.
\end{lemma}

\begin{proof}
It is easy to show that $N_{U_{2}}(n)=n+1$ for any $n\in \mathbb{N}$. Using
 Proposition \ref{P5} we obtain $L_{U_{2}}^{d}(n)=L_{U_{2}}^{a}(n)=2(n+1)$
for any $n\in \mathbb{N}$. Let $n\in \mathbb{N}$. Choose $t\in \mathbb{N}$
such that $2^{t}>n$. Consider a problem $z=(\nu ,p_{2^{t}+1},\ldots
,p_{2^{t}+n})$ over $U_{2}$ such that, for each $\bar{\delta}_{1},\bar{\delta%
}_{2}\in \{0,1\}^{n}$ with $\bar{\delta}_{1}\neq \bar{\delta}_{2}$, $\nu (%
\bar{\delta}_{1})\neq \nu (\bar{\delta}_{2})$. Consider an arbitrary
decision tree $\Gamma $ over $U_{2}$ that solves the problem $z$
deterministically and a complete path $\xi $ of $\Gamma $ such that $2^t+n+1 \in A_2(\xi)$. One can show
that if the number of working nodes in $\xi $ is less than $n$, then the
function $z(x)$ is not constant on the set $A_{2}(\xi )$ but this is
impossible. Therefore $h(\Gamma )\geq n$ and $h_{U_{2}}^{d}(n)\geq n$. It is
clear that $h_{U_{2}}^{d}(n)\leq n$. Hence $h_{U_{2}}^{d}(n)=n$. It is easy
to show that $h_{U_{2}}^{a}(n)\geq 1$. We now show that $h_{U_{2}}^{a}(n)%
\leq 1$. Let $z=(\nu ,f_{1},\ldots ,f_{n})$ be an arbitrary problem over $%
U_{2}$. Each attribute from the set $\{f_{1},\ldots ,f_{n}\}$ is of the kind
$p_{i}$ or $l_{i}$. We say about the number $i$ as about the index of the
considered attribute. Let $j$ be the maximum index of an attribute from the
set $\{f_{1},\ldots ,f_{n}\}$ and $t$ be a number from $\mathbb{N}$ such
that $2^{t}>j$. Then the function $z(x)$ is constant on the sets of
solutions of equation systems $\{p_{1}(x)=1\},\ldots ,\{p_{2^{t}}(x)=1\},$
$\{l_{2^{t}}(x)=1\}$ on $A_{2}$, and the union of these sets of solutions is
equal to $A_{2}$. Using these facts it is easy to show that there exists a
decision tree $\Gamma $ over $U_{2}$ which solves the problem $z$
nondeterministically and for which $h(\Gamma )=1$. Therefore $%
h_{U_{2}}^{a}(n)\leq 1$. Hence $h_{U_{1}}^{a}(n)=1$.

Since the function $%
h_{U_{2}}^{d}$ has the type of behavior \textrm{LIN}, the function $%
h_{U_{2}}^{a}$ has the type of behavior \textrm{CON}, and the functions $%
L_{U_{2}}^{d}$ and $L_{U_{2}}^{a}$ have the type of behavior \textrm{POL}
the information system $U_{2}$ belongs to the class $W_{2}$ -- see  Table \ref{tab1}.
One can show that the information system $U_{2}$ satisfies the condition of
coverage with  parameter $2$ and has finite I-dimension equals to $1$.
Using  Proposition \ref{P1} we obtain that this information system does not
satisfy the condition of restricted coverage.
\end{proof}

Define an information system $U_{3}=(A_{3},F_{3})$ as follows: $A_{3}=%
\mathbb{N}$ and $F_{3}=\{p_{i}:i\in \mathbb{N}\}$.

\begin{lemma}
\label{L6.3}The information system $U_{3}$ belongs to the class $W_{3}$, $%
h_{U_{3}}^{d}(n)=n$, $h_{U_{3}}^{a}(n)=n$, $L_{U_{3}}^{d}(n)=2(n+1)$, and $%
L_{U_{3}}^{a}(n)=2(n+1)$ for any $n\in \mathbb{N}$. This information system
does not satisfy the conditions of coverage and restricted coverage, and has
finite I-dimension equals to $1$.
\end{lemma}

\begin{proof}
It is easy to show that $N_{U_{3}}(n)=n+1$ for any $n\in \mathbb{N}$. Using
 Proposition \ref{P5} we obtain $L_{U_{3}}^{d}(n)=L_{U_{3}}^{a}(n)=2(n+1)$
for any $n\in \mathbb{N}$. Let $n\in \mathbb{N}$. Consider a problem $z=(\nu
,p_{1},\ldots ,p_{n})$ over $U_{3}$ such that, for each $\bar{\delta}_{1},%
\bar{\delta}_{2}\in \{0,1\}^{n}$ with $\bar{\delta}_{1}\neq \bar{\delta}_{2}$%
, $\nu (\bar{\delta}_{1})\neq \nu (\bar{\delta}_{2})$. Consider an arbitrary
decision tree $\Gamma _{1}$ over $U_{3}$ that solves the problem $z$
deterministically and a complete path $\xi $ in $\Gamma _{1}$ in which each
edge leaving a working node is labeled with the number $0$. One can show
that if the number of working nodes in $\xi $ is less than $n$, then the
function $z(x)$ is not constant on the set $A_{3}(\xi )$ but this is
impossible. Therefore $h(\Gamma _{1})\geq n$ and $h_{U_{3}}^{d}(n)\geq n$.
It is clear that $h_{U_{3}}^{d}(n)\leq n$. Hence $h_{U_3}^{d}(n)=n$. Consider
an arbitrary decision tree $\Gamma _{2}$ over $U_{3}$ that solves the
problem $z$ nondeterministically. Let $\xi _{1}$ be a complete path of $%
\Gamma _{2}$ in which at least one edge leaving a working node is labeled
with the number $1$. Then the set $A_{3}(\xi _{1})$ contains at most one
element from $A_{3}$. The set $A_{3}$ is infinite, the number of complete
paths in $\Gamma _{2}$ is finite, and the union of the sets $A_{3}(\xi )$ for
all complete paths $\xi $ in $\Gamma $ is equal to $A_{3}$. Therefore there
exists a complete path $\xi _{0}$ in $\Gamma _{2}$ in which each edge leaving
a working node is labeled with the number $0$. One can show that if the
number of working nodes in $\xi _{0}$ is less than $n$, then the function $%
z(x)$ is not constant on the set $A_{3}(\xi _{0})$ but this is impossible.
Therefore $h(\Gamma _{2})\geq n$ and $h_{U_{3}}^{a}(n)\geq n$. It is clear
that $h_{U_{3}}^{a}(n)\leq n$. Hence $h_{U_{3}}^{a}(n)=n$.

Since the
functions $h_{U_{3}}^{d}$ and $h_{U_{3}}^{a}$ have the type of behavior
\textrm{LIN} and the functions $L_{U_{3}}^{d}$ and $L_{U_{3}}^{a}$ have the
type of behavior \textrm{POL}, the information system $U_{3}$ belongs to the
class $W_{3}$ -- see  Table \ref{tab1}. One can show that this information system has
finite I-dimension equals to $1$. Using Propositions \ref{P1} and \ref{P2}
we obtain that the information system $U_{3}$ does not satisfy the
conditions of coverage and restricted coverage.
\end{proof}

Define an information system $U_{4}=(A_{4},F_{4})$ as follows: $A_{4}=%
\mathbb{N}$ and $F_{4}$ is the set of all functions from $\mathbb{N}$ to $%
\{0,1\}$.

\begin{lemma}
\label{L6.4}The information system $U_{4}$ belongs to the class $W_{4}$, $%
h_{U_{4}}^{d}(n)=n$, $h_{U_{4}}^{a}(n)=1$, $L_{U_{4}}^{d}(n)=2^{n+1}$, and $%
L_{U_{4}}^{a}(n)=2^{n+1}$ for any $n\in \mathbb{N}$. This information system
satisfies the condition of coverage with parameter $2$, satisfies the
condition of restricted coverage with parameters $2$ and $1$, and has
infinite I-dimension.
\end{lemma}

\begin{proof}
It is easy to show that the information system $U_{4}$ has infinite $I$%
-dimension. By  Proposition \ref{P6}, $N_{U_{4}}(n)=2^{n}$ for any $n\in
\mathbb{N}$. Using Proposition \ref{P5} we obtain $%
L_{U_{4}}^{d}(n)=L_{U_{4}}^{a}(n)=2^{n+1}$ for any $n\in \mathbb{N}$. Let $%
n\in \mathbb{N}$. By  Proposition \ref{P1}, $h_{U_{4}}^{d}(n)=n$. It is easy
to show that $h_{U_{4}}^{a}(n)\geq 1$. We now show that $h_{U_{4}}^{a}(n)%
\leq 1$. Let $z=(\nu ,f_{1},\ldots ,f_{n})$ be an arbitrary problem over $%
U_{4}$. Let $(\delta _{1},\ldots ,\delta _{n})\in \{0,1\}^{n}$ and the set $%
B $ of solutions from $A_{4}$ of the equation system $\{f_{1}(x)=\delta
_{1},\ldots ,f_{n}(x)=\delta _{n}\}$ is nonempty. Then there exists a
function $f\in F_{4}$ such that the set $B$ is the set of solutions from $%
A_{4}$ of the equation system $\{f(x)=1\}$. Using this fact it is easy to
show that there exists a decision tree $\Gamma $ over $U_{4}$ which solves
the problem $z$ nondeterministically and for which $h(\Gamma )=1$. Therefore
$h_{U_{4}}^{a}(n)\leq 1$. Hence $h_{U_{4}}^{a}(n)=1$.

Since the function $%
h_{U_{4}}^{d}$ has the type of behavior \textrm{LIN}, the function $%
h_{U_{4}}^{a}$ has the type of behavior \textrm{CON}, and the functions $%
L_{U_{4}}^{d}$ and $L_{U_{4}}^{a}$ have the type of behavior \textrm{EXP},
the information system $U_{4}$ belongs to the class $W_{4}$ -- see  Table \ref{tab1}.
One can show that this information system satisfies the condition of
coverage with parameter $2$ and satisfies the condition of restricted
coverage with parameters $2$ and $1$.
\end{proof}

Define an information system $U_{5}=(A_{5},F_{5})$ as follows: $A_{5}$ is
the set of all infinite sequences $a_{1},a_{2},\ldots $, where $a_{i}\in
\{0,1\}$ for any $i\in \mathbb{N}$, and $F_{5}=\{f_{i}:i\in \mathbb{N}\}$,
where $f_{i}(a_{1},a_{2},\ldots )=a_{i}$ for any $a_{1},a_{2},\ldots \in
A_{5}$ and $i\in \mathbb{N}$.

\begin{lemma}
\label{L6.5}The information system $U_{5}$ belongs to the class $W_{5}$, $%
h_{U_{5}}^{d}(n)=n$, $h_{U_{5}}^{a}(n)=n$, $L_{U_{5}}^{d}(n)=2^{n+1}$, and $%
L_{U_{5}}^{a}(n)=2^{n+1}$ for any $n\in \mathbb{N}$. This information system
does not satisfy the condition of coverage and the condition of restricted
coverage, and has infinite I-dimension.
\end{lemma}

\begin{proof}
Let $n\in \mathbb{N}$. One can show that the equation system $%
\{f_{1}(x)=\delta _{1},\ldots ,f_{n}(x)=\delta _{n}\}$ has a solution from
the set $A_{5}$ for any $(\delta _{1},\ldots ,\delta _{n})\in \{0,1\}^{n}$.
Therefore the information system $U_{5}$ has infinite I-dimension. By
 Proposition \ref{P6}, $N_{U_{5}}(n)=2^{n}$ for any $n\in \mathbb{N}$. Using
 Proposition \ref{P5} we obtain $L_{U_{5}}^{d}(n)=L_{U_{5}}^{a}(n)=2^{n+1}$
for any $n\in \mathbb{N}$. Let $n\in \mathbb{N}$. By  Proposition \ref{P1}, $%
h_{U_{5}}^{d}(n)=n$. Consider a problem $z=(\nu ,f_{1},\ldots ,f_{n})$ over $%
U_{4}$ such that, for each $\bar{\delta}_{1},\bar{\delta}_{2}\in \{0,1\}^{n}$
with $\bar{\delta}_{1}\neq \bar{\delta}_{2}$, $\nu (\bar{\delta}_{1})\neq
\nu (\bar{\delta}_{2})$. Let $\Gamma $ be an arbitrary decision tree over $%
U_{5}$ that solves the problem $z$ nondeterministically and $\xi $ be an
arbitrary complete path of $\Gamma $. We denote by $G$ the set of attributes
attached to working nodes of $\xi $. One can show that, if $\{f_{1},\ldots
,f_{n}\}$ is not a subset of the set $G$, then the function $z(x)$ is not
constant on the set $A_{5}(\xi )$ but this is impossible. Therefore $%
h(\Gamma )\geq n$ and $h_{U_{5}}^{a}(n)\geq n$. It is clear that $%
h_{U_{5}}^{a}(n)\leq n$. Thus, $h_{U_{5}}^{a}(n)=n$.

Since the functions $%
h_{U_{5}}^{d}$ and $h_{U_{5}}^{a}$ have the type of behavior \textrm{LIN}
and the functions $L_{U_{5}}^{d}$ and $L_{U_{5}}^{a}$ have the type of
behavior \textrm{EXP}, the information system $U_{5}$ belongs to the class $%
W_{5}$ -- see  Table \ref{tab1}. By Proposition  \ref{P2} this information system does
not satisfy the condition of coverage and therefore does not satisfy the
condition of restricted coverage.
\end{proof}

\begin{proof}
[Proof of Theorem \ref{T1}] The statements of the theorem follow from Lemmas \ref%
{L5}-\ref{L6.5}.
\end{proof}

\section{Proofs of Theorems \ref{T2}-\ref{T6}} \label{S5}

First, we prove seven auxiliary statements.

\begin{lemma}
\label{L7}Let $U$ be an infinite binary information system such that $%
h_{U}^{d}(n)=n$ for any $n\in \mathbb{N}$. Then the information system $U$
is $d$-reachable.
\end{lemma}

\begin{proof}
Let $z=(\nu ,f_{1},\ldots ,f_{n})$ be a problem over $U$ and $\Gamma $ be a
decision tree that solves the problem $z$ deterministically and satisfies
the following conditions: the number of working nodes in each complete path
of $\Gamma $ is equal to $n$ and these nodes in the order from the root to a
terminal node are labeled with attributes $f_{1},\ldots ,f_{n}$. Remove from $%
\Gamma $ all nodes and edges that do not belong to realizable complete
paths. Let $w$ be a working node in the obtained tree that has only one
leaving edge $d$ entering a node $v$. We remove the node $w$ and edge $d$
and connect the edge entering $w$ to the node $v$. We do the same with all
working nodes with only one leaving edge. Denote by $\Gamma ^{\prime }$ the
obtained decision tree. It is clear that $\Gamma ^{\prime }$ solves the
problem $z$ deterministically, $\Gamma ^{\prime }\in G_{d}^{2}(U)$, and $%
L_{t}(\Gamma ^{\prime })\leq N_{U}(f_{1},\ldots ,f_{n})\leq N_{U}(n)$. By
 Lemma \ref{L3}, $L_{w}(\Gamma ^{\prime })=L_{t}(\Gamma ^{\prime })-1$.
Therefore $L(\Gamma ^{\prime })\leq 2N_{U}(n)$. Using  Proposition \ref{P5}
we obtain $L(\Gamma ^{\prime })\leq L_{U}^{d}(n)$. It is clear that $h(\Gamma
^{\prime })\le n=h_{U}^{d}(n)$. Therefore $U$ is $d$-reachable.
\end{proof}

\begin{lemma}
\label{L8}Let $U$ be an infinite binary information system such that $%
h_{U}^{a}(n)=n$ for any $n\in \mathbb{N}$. Then the information system $U$
is $a$-reachable.
\end{lemma}

\begin{proof}
Let $z=(\nu ,f_{1},\ldots ,f_{n})$ be a problem over $U$. Construct for this
problem the decision tree $\Gamma ^{\prime }$ as in the proof of  Lemma \ref%
{L7}. It is clear that $\Gamma ^{\prime }$ solves the problem $z$
nondeterministically. We know that $L(\Gamma ^{\prime })\leq 2N_{U}(n)$.
Using  Proposition \ref{P5} we obtain $L(\Gamma ^{\prime })\leq L_{U}^{a}(n)$
It is clear that $h(\Gamma ^{\prime })\le n=h_{U}^{a}(n)$. Therefore $U$ is $a$%
-reachable.
\end{proof}

\begin{lemma}
\label{L9}Let $U$ be an infinite binary information system which satisfies
the condition of coverage. Then the information system $U$ is not $a$-reachable.
\end{lemma}

\begin{proof}
By Proposition \ref{P2}, the function $h_{U}^{a}(n)$ is bounded from above
by a positive constant $c$. By  Proposition \ref{P6}, the function $N_{U}(n)$
is not bounded from above by a constant. Choose $n\in \mathbb{N}$ such that $%
N_{U}(n)>2^{2c}$. Let $z=(\nu ,f_{1},\ldots ,f_{n})$ be a problem over $U$ such that $\nu (\bar{\delta}_{1})\neq \nu (\bar{\delta}_{2})$ for
any $\bar{\delta}_{1},\bar{\delta}_{2}\in \{0,1\}^{n}$, $\bar{\delta}%
_{1}\neq \bar{\delta}_{2}$, and $N_{U}(f_{1},\ldots ,f_{n})=N_{U}(n)$. Let $%
\Gamma $ be a decision tree over $U$ which solves the problem $z$
nondeterministically, for which $h(\Gamma )\leq h_{U}^{a}(n)\leq c$, and
which has the minimum number of nodes among such trees. In the same way as it
was done in the proof of  Proposition \ref{L2}, we can prove that $\Gamma \in
G_{a}^{f}(U)$. It is clear that $L_{t}(\Gamma )\geq N_{U}(f_{1},\ldots
,f_{n})=N_{U}(n)$. Let us assume that $\Gamma \in G_{d}^{2}(U)$. Then it is
easy to show that $h(\Gamma )\geq \log _{2}L_{t}(\Gamma )\geq \log
_{2}N_{U}(n)>2c$ which is impossible. Therefore $\Gamma \in
G_{a}^{f}(U)\setminus G_{d}^{2}(U)$. By  Lemma \ref{L4}, $L_{w}(\Gamma
)>L_{t}(\Gamma )-1\geq N_{U}(n)-1$. Using  Proposition \ref{P5} we obtain $%
L(\Gamma )>2N_{U}(n)=L_{U}^{a}(n)$. Therefore $U$ is not $a$-reachable.
\end{proof}

\begin{lemma}
\label{L10}Let $U$ be an infinite binary information system which does not
satisfy the condition of restricted coverage. Then $$(n,L_{U}^{a}(n))$$ is the
optimal boundary $a$-pair of the information system $U$.
\end{lemma}

\begin{proof}
The proof of the fact that $(n,L_{U}^{a}(n))$ is a boundary $a$-pair of the
system $U$ is similar to the proof of  Lemma \ref{L8}.

We now prove that this is the optimal boundary $a$-pair of the system $U$.
Let $(q,r)$ be a boundary $a$-pair of the system $U$. It is clear that $%
r(n)\geq L_{U}^{a}(n)$ for any $n\in \mathbb{N}$. Let us show that $q(n)\geq
n$ for any $n\in \mathbb{N}$. Assume the contrary: there exists $m\in
\mathbb{N}$ such that $q(m+1)\leq m$. Let us consider an arbitrary $(m+1,U)$%
-set $B$ given by $(m+1,U)$-system of equations
\[
\{f_{1}(x)=\delta _{1},\ldots ,f_{m+1}(x)=\delta _{m+1}\}.
\]%
Consider the problem $z=(\nu ,f_{1},\ldots ,f_{m+1})$ over $%
U$ such that, for any $\bar{\delta}\in \{0,1\}^{m+1}$, $\nu (\bar{\delta})\in
\{1,2\}$ and $\nu (\bar{\delta})=1$ if and only if $\bar{\delta}=(\delta
_{1},\ldots ,\delta _{m+1})$. Let $\Gamma $ be a decision tree over $U$ which
solves the problem $z$ nondeterministically and for which $h(\Gamma )\leq
q(m+1)\leq m$ and $L(\Gamma )\leq r(m+1)$. Let $\xi _{1},\ldots ,\xi _{t}$
be all realizable complete paths of $\Gamma $ in which terminal nodes are
labeled with the number $1$. It is clear that $A(\xi _{1}),\ldots ,A(\xi
_{t})$ are $(m,U)$-sets, $B=A(\xi _{1})\cup \ldots \cup A(\xi _{t})$, and $%
t\leq r(m+1)$. Since $B$ is an arbitrary $(m+1)$-set, we obtain that $U$
satisfies the condition of restricted coverage but this is impossible.
Therefore $(n,L_{U}^{a}(n))$ is the optimal boundary $a$-pair of the system $U$.
\end{proof}

Let $\mathbb{R}_{+}$ be the set of nonnegative real numbers. Define an infinite
binary information system $U_{6}=(A_{6},F_{6})$ as follows: $A_{6}=\mathbb{R}%
_{+}$ and $F_{6}=\{p_{i}:i\in \mathbb{N}\}\cup \{q_{i}:i\in \mathbb{N}\}$,
where, for any $a\in A_{6}$, $p_{i}(a)=1$ if and only if $a=i$, and $%
q_{i}(a)=1$ if and only if $a\geq i+\frac{1}{2}$.

\begin{lemma}
\label{L11}The information system $U_{6}$ belongs to the class $W_{1}$. This
information system is not $d$-reachable.
\end{lemma}

\begin{proof}
It is easy to show that $N_{U_{6}}(n)=n+1$ for any $n\in \mathbb{N}$. Using
 Proposition \ref{P5} we obtain $L_{U_{6}}^{d}(n)=2(n+1)$ for any $n\in
\mathbb{N}$. We now show that $h_{U_{6}}^{d}(n)\leq \lceil \log
_{2}(n+1)\rceil +1$ for any $n\in \mathbb{N}$. Let $z=(\nu ,f_{1},\ldots
,f_{n})$ be an arbitrary problem over $U_{6}$. We  describe a decision
tree $\Gamma $ over $U_{6}$ that solves the problem  $z$ deterministically. This tree
will use attributes from a set $G$ containing all attributes $f_{1},\ldots
,f_{n}$ and, probably, some additional attributes from the set $\{q_{i}:i\in
\mathbb{N}\}$. Initially, $G=\{f_{1},\ldots ,f_{n}\}$. Let $p_{i_{1}},\ldots
,p_{i_{m}}$ be all attributes from the set $\{p_{i}:i\in \mathbb{N}\}\cap
\{f_{1},\ldots ,f_{n}\}$ ordered such that $i_{1}<\ldots <i_{m}$. For $%
t=1,\ldots ,m-1$, if the set $G$ does not contain any attribute $q_{j}$ such
that $i_t\leq j<i_{t+1}$, then we add to $G$ the attribute $q_{i_t}$. As a result,
we obtain a set of attributes $G$ that contains at most $n$ attributes from
the set $\{q_{i}:i\in \mathbb{N}\}$. Let there are attributes $%
q_{j_{1}},\ldots ,q_{j_{k}}$, $k \le n$. It is easy to construct a decision tree $%
\Gamma ^{\prime }$ that finds values of these attributes on a given element $%
a\in A_{6}$ and has depth at most $\lceil \log _{2}(n+1)\rceil $. If we know
the values of all attributes $q_{j_{1}},\ldots ,q_{j_{k}}$, then we know
values of all attributes $p_{i_{1}},\ldots ,p_{i_{m}}$ on $a$ with the
exception of at most one attribute. If we compute the value of this
attribute we will know the values of all attributes $f_{1},\ldots ,f_{n}$ on
$a$ and the value $z(a)$. So we can transform the decision tree $\Gamma
^{\prime }$ into a deterministic decision tree $\Gamma $ over $U_{6}$ which solves the
problem $z$ and which depth is at most $\lceil \log _{2}(n+1)\rceil +1$.
Therefore $h_{U_{6}}^{d}(n)\leq \lceil \log _{2}(n+1)\rceil +1$. Since the function $h_{U_{6}}^{d}$ has the type of behavior  \textrm{LOG}, the information system $U_6$ belongs to the class $W_1$ -- see Table \ref{tab1}.

We now show that the information system $U_{6}$ is not $d$-reachable. Assume
the contrary. Choose natural $n$ such that $\lceil \log _{2}(n+1)\rceil +1<n$%
. Consider a problem $z=(\nu ,p_{1},\ldots ,p_{n})$ such that, for any $\bar{%
\delta}_{1},\bar{\delta}_{2}\in \{0,1\}^{n}$, if $\bar{\delta}_{1}\neq \bar{%
\delta}_{2}$, then $\nu (\bar{\delta}_{1})\neq \nu (\bar{\delta}_{2})$. Let
for the definiteness, $z(i)=i$ for $i=1,\ldots ,n$ and $z(a)=n+1$ for any $a\in
A_{6}\setminus \{1,\ldots ,n\}$. According to the assumption, there exists a deterministic
decision tree $\Gamma $ over $U_{6}$ which solves the problem $z$ and for
which $h(\Gamma )\leq \lceil \log _{2}(n+1)\rceil +1$ and $L(\Gamma )\leq
2(n+1)$. It is clear, that for each $i\in \{1,\ldots ,n,n+1\}$, there is at
least one complete path $\xi $ of $\Gamma $ in which the terminal node is
labeled with the number $i$. Let us assume that there exists exactly one
complete path $\xi $ of $\Gamma $ in which the terminal node is labeled with
the number $n+1$. In this case, $A_{6}(\xi )=A_{6}\setminus \{1,\ldots ,n\}$.
It means that, for each $i\in \{1,\ldots ,n\}$, there exists a working node
of $\xi $ that is labeled with the attribute $p_{i}$. Hence $h(\Gamma
)\geq n$ which is impossible. Therefore, $L_{t}(\Gamma )\geq n+2$. Using
 Lemma \ref{L3} we obtain $L_{w}(\Gamma )\geq L_{t}(\Gamma )-1\geq n+1$. Thus,
$L(\Gamma )\geq 2(n+2)$ but this is impossible. Therefore $U_{6}$ is not $d$%
-reachable.
\end{proof}

Let $\mathbb{Z}$ be the set of integers and $\mathbb{Z}_{-}=\mathbb{Z}%
\setminus \mathbb{N}$. Define an information system $U_{7}=(A_{7},F_{7})$ as
follows: $A_{7}=\mathbb{Z}$ and $F_{7}=\{p_{i}:i\in \mathbb{N}\}\cup
\{l_{2^{i}}:i\in \mathbb{N}\}\cup F_{0}$. The functions $p_{i}$ and $%
l_{2^{i}}$ have value $0$ on the set $Z_{-}$, and $F_{0}$ is the set of all
functions from $\mathbb{Z}$ to $\{0,1\}$ that are equal to $0$ on the set $%
\mathbb{N}$.

\begin{lemma}
\label{L12}The information system $U_{7}$ belongs to the class $W_{4}$ and
does not satisfy the condition of restricted coverage.
\end{lemma}

\begin{proof}
Let $n\in \mathbb{N}$. One can show that there are attributes $g_{1},\ldots
,g_{n}\in F_{0}$ such that $N_{U_{7}}(g_{1},\ldots ,g_{n})=2^{n}$. Therefore
the information system $U_{7}$ has infinite I-dimension. Using  %
Proposition \ref{P3} we obtain that the function $L_{U_{7}}^{d}$ has the type of
behavior \textrm{EXP}. We now show that $h_{U_{7}}^{a}(n)\leq 1$.
Let $z=(\nu ,f_{1},\ldots ,f_{n})$ be an arbitrary problem over $U_{7}$. The
attributes $f_{1},\ldots ,f_{n}$ divide the set $A_{7}$ into finite number
of nonempty domains in each of which these attributes have fixed values. One
can show that each domain can be represented as a union of finite number of
subdomains such that each subdomain is the set of solutions on $A_{7}$ of
equation system of the kind $\{f(x)=1\}$, where $f\in F_{0}$, or $%
\{p_{i}(x)=1\}$, or $\{l_{2^{i}}(x)=1\}$. Using these facts it is easy to
show that there exists a decision tree $\Gamma $ over $U_{7}$ which solves
the problem $z$ nondeterministically and for which $h(\Gamma )=1$. Therefore
$h_{U_{7}}^{a}(n)\leq 1$. Hence the function $h_{U_{7}}^{a}$ has the type of
behavior \textrm{CON}. Taking into account that the function $L_U^d$ has the type of behavior \textrm{EXP}, we obtain that the information system $U_7$ belongs to the class $W_{4}$ -- see Table  \ref{tab1}.

Let us assume that the information system $U_{7}$ satisfies the condition of
restricted coverage with parameters $m$ and $t$. Let $r$ $be$ a natural
number such that $t<2^{r}-m+1$. Let $B$ be the set of
solutions on $A_{7}$ of the equation system
\[
\{l_{2^{r}}(x)=1,p_{2^{r}+1}(x)=0,\ldots
,p_{2^{r}+m-1}(x)=0,l_{2^{r+1}}(x)=0\}.
\]%
It is clear that $B=\{2^{r}+m,2^{r}+m+1,\ldots ,2^{r+1}\}$ and $%
|B|=2^{r}-m+1 $. Let us assume that $B$ is a union of at most $t$ $(U_{7},m)$%
-sets. Let $C$ be a $(U_{7},m)$-set such that $C\subseteq B$. One can show
that $|C|=1$. Therefore $t\geq 2^{r}-m+1$ but this is impossible. Therefore
the information system $U_{7}$ does not satisfy the condition of restricted
coverage.
\end{proof}

\begin{lemma}
\label{L13}Let $U=(A,F)$ be an information system from the class $W_{4}$ which
satisfies the condition of restricted coverage. Then there exist positive
constants $c_{1}$ and $c_{2}$ such that $(c_{1},c\,_{2}^{n})$ is a boundary $%
a$-pair of the system $U$.
\end{lemma}

\begin{proof}
Let $U$ satisfy the condition of restricted coverage with parameters $m$
and $t$: any $(m+1,U)$-set is a union of at most $t$ $(m,U)$-sets. From
 Lemma \ref{L0} it follows that, for any $n\in \mathbb{N}$, any $(n,U)$-set
is a union of at most $t^{n}$ $(m,U)$-sets.

Let $z=(\nu ,f_{1},\ldots ,f_{n})$ be a problem over $U$. We now show that
there exists a decision tree $\Gamma $ over $U$ which solves the problem $z$
nondeterministically and for which $h(\Gamma )\leq m$ and $L(\Gamma )\leq (m+2)2^nt^n$.
Let $\bar{\delta}=(\delta _{1},\ldots ,\delta _{n})$ be a tuple from $%
\{0,1\}^{n}$ such that the equation system
\[
S(\bar{\delta})=\{f_{1}(x)=\delta _{1},\ldots ,f_{n}(x)=\delta _{n}\}
\]%
has a solution from $A$. For each solution $a\in A$ of this system, we have $%
z(a)=\nu (\bar{\delta})$. The set of solutions of $S(\bar{\delta})$ is a
union of $(m,U)$-sets $D_{1},\ldots ,D_{s}$, where $s\leq t^{n}$. Each of
these sets $D_{i}$ is the set of solutions of an $(m,U)$-system of
equations. For the considered $(m,U)$-system of
equations, we construct a complete path $\xi _{i}$ with $m$ working
nodes such that $A(\xi _{i})=D_{i}$ and the terminal node of $\xi _{i}$ is
labeled with the number $\nu (\bar{\delta})$. Denote $\Sigma (\bar{\delta}%
)=\{\xi _{1},\ldots ,\xi _{s}\}$. Let $\Sigma =\bigcup \Sigma (\bar{\delta}%
) $, where the union is considered among all tuples $\bar{\delta}\in
\{0,1\}^{n}$ such that the set of equations $S(\bar{\delta})$ has a solution
from $A$ (the number of such tuples is at most $2^{n})$. We identify initial
nodes of all paths from $\Sigma $. As a result, we obtain a decision tree $%
\Gamma $ over the information system $U$ which solves the problem $z$
nondeterministically and for which $h(\Gamma )\leq m$ and $L(\Gamma )\leq
(m+2)2^{n}t^{n}$. Denote $c_{1}=m$ and $c_{2}=(m+2)2t$. Then $h(\Gamma )\leq
c_{1}$ and $L(\Gamma )\leq c_{2}^{n}$. Taking into account that $z$ is an
arbitrary problem over $U$ with $\dim z =n$, we obtain that $(c_{1},c\,_{2}^{n})$ is a
boundary $a$-pair of the system $U$.
\end{proof}

\begin{proof}
[Proof of Theorem \ref{T2}] Each information system from the class $W_{1}$
satisfies the condition of coverage and the condition of restricted coverage,
and has finite I-dimension (see Table \ref{tab2}).

(a) The existence of both $d$-reachable and not $d$-reachable information
systems in the class $W_{1}$ follows from Lemmas  \ref{L6.1} and \ref{L11}.
Let $U$ be an information system from $W_{1}$ which is not $d$-reachable.
Then $U$ satisfies the condition of restricted coverage and has finite $I$%
-dimension. From here and from Theorem 2.1 \cite{Moshkov03} it follows that, for
each $\varepsilon $, $0<\varepsilon <1$, there exists a positive constant $c$
such that $(c(\log _{2}n)^{1+\varepsilon }+1,2^{c(\log _{2}n)^{1+\varepsilon
}+2})$ is a boundary $d$-pair of the system $U$.

(b) Let $U=(A,F)$ be an information system from the class $W_{1}$. Then $U$
satisfies the condition of coverage. Using Lemma \ref{L9} we obtain that the
system $U$ is not $a$-reachable. We know that the system $U$ satisfies the
condition of restricted coverage and has finite I-dimension. Using
Theorem 2.2 \cite{Moshkov03} we obtain that, for each $\varepsilon $, $0<\varepsilon
<1$, there exist positive constants $c_{1}$, $c_{2}$ such that, for each
problem $z$ over $U$, there exists a system $\Delta $ of decision rules of
the kind
\[
(f_{1}(x)=\delta _{1})\wedge \ldots \wedge (f_{m}(x)=\delta _{m})\rightarrow
z(x)=\sigma ,
\]%
where $f_{1},\ldots ,f_{m}\in F$, $\delta _{1},\ldots ,\delta _{m}\in
\{0,1\} $, and $\sigma \in \mathbb{N}$, that satisfies the following
conditions: (i) each rule is true for the problem $z$, (ii) for each $a$ $%
\in A$, there exists a rule from $\Delta $ that accepts $a$, (iii) the
number of conditions in the left-hand side of each rule is at most $c_{1}$,
and (iv) the number of riles in $\Delta $ is at most $2^{c_{2}(\log
_{2}n)^{1+\varepsilon }+1}$ where $n=\dim z$. One can transform the system
of decision rules $\Delta $ into a decision tree $\Gamma $ over $U$ which
solves the problem $z$ nondeterministically and for which $h(\Gamma )\leq
c_{1}$ and $L(\Gamma )\leq (c_{1}+2)2^{c_{2}(\log _{2}n)^{1+\varepsilon
}+1}=  2^{c_{2}(\log _{2}n)^{1+\varepsilon }+c_{3}}$, where $c_{3}=\log_2(c_1+2)+1$.
Note that complete paths
in $\Gamma $ correspond to rules from the system $\Delta $. Thus, $%
(c_{1},2^{c_{2}(\log _{2}n)^{1+\varepsilon }+c_{3}})$ is a boundary $a$-pair
of the system $U$.
\end{proof}

\begin{proof}
[Proof of Theorem \ref{T3}] Each information system from the class $W_{2}$
satisfies the condition of coverage, does not satisfy the condition of
restricted coverage, and has finite I-dimension (see Table  \ref{tab2}).

(a) Let $U$ be an information system from the class $W_{2}$. Then $U$ does
not satisfy the condition of restricted coverage. By Proposition  \ref{P1}, $%
h_{U}^{d}(n)=n$ for any $n\in \mathbb{N}$. Using Lemma  \ref{L7} we obtain
that the system $U$ is $d$-reachable.

(b) Let $U$ be an information system from the class $W_{2}$. Then $U$
satisfies the condition of coverage. Using Lemma \ref{L9} we obtain that the
system $U$ is not $a$-reachable. We know that $U$ does not satisfy the
condition of restricted coverage. Using Lemma \ref{L10} we obtain that $%
(n,L_{U}^{a}(n))$ is the optimal boundary $a$-pair of the system $U$.
\end{proof}

\begin{proof}
[Proof of Theorem \ref{T4}] Each information system from the class $W_{3}$ does
not satisfy the condition of coverage and the condition of restricted
coverage, and has finite I-dimension (see Table  \ref{tab2}).

(a) Let $U$ be an information system from the class $W_{3}$. Then $U$ does
not satisfy the condition of restricted coverage. By Proposition \ref{P1}, $%
h_{U}^{d}(n)=n$ for any $n\in \mathbb{N}$. Using Lemma  \ref{L7} we obtain
that the system $U$ is $d$-reachable.

(b) Let $U$ be an information system from the class $W_{3}$. Then $U$ does
not satisfy the condition of coverage. By  Proposition \ref{P2}, $%
h_{U}^{a}(n)=n$ for any $n\in \mathbb{N}$. Using Lemma \ref{L8} we obtain
that the system $U$ is $a$-reachable.
\end{proof}

\begin{proof}
[Proof of Theorem \ref{T5}] Each information system from the class $W_{4}$
satisfies the condition of coverage and has infinite I-dimension (see Table
 \ref{tab2}). From Lemmas  \ref{L6.4} and \ref{L12} it follows that the class $%
W_{4}$ contains both information systems that satisfy the
condition of restricted coverage and information systems that do not satisfy this condition.

(a) Let $U$ be an information system from the class $W_{4}$. Then $U$ has
infinite I-dimension. By  Proposition \ref{P1}, $h_{U}^{d}(n)=n$ for any $%
n\in \mathbb{N}$. Using  Lemma \ref{L7} we obtain that the system $U$ is $d$%
-reachable.

(b) Let $U$ be an information system from the class $W_{4}$. Then $U$
satisfies the condition of coverage. Using  Lemma \ref{L9} we obtain that the
system $U$ is not $a$-reachable. Let $U$ do not satisfy the condition of
restricted coverage. Using  Lemma \ref{L10} we obtain that $(n,L_{U}^{a}(n))$ is
the optimal boundary $a$-pair of the system $U$. Let $U$ satisfy the
condition of restricted coverage. From  Lemma \ref{L13} it follows that there
exist positive constants $c_{1}$ and $c_{2}$ such that $(c_{1},c\,_{2}^{n})$
is a boundary $a$-pair of the system $U$.
\end{proof}

\begin{proof}
[Proof of Theorem \ref{T6}] Each information system from the class $W_{5}$ does
not satisfy the condition of coverage and the condition of restricted
coverage, and has infinite I-dimension (see  Table \ref{tab2}).

(a) Let $U$ be an information system from the class $W_{5}$. Then $U$ does
not satisfy the condition of restricted coverage. By  Proposition \ref{P1}, $%
h_{U}^{d}(n)=n$ for any $n\in \mathbb{N}$. Using  Lemma \ref{L7} we obtain
that the system $U$ is $d$-reachable.

(b) Let $U$ be an information system from the class $W_{5}$. Then $U$ does
not satisfy the condition of coverage. By  Proposition \ref{P2}, $%
h_{U}^{a}(n)=n$ for any $n\in \mathbb{N}$. Using  Lemma \ref{L8} we obtain
that the system $U$ is $a$-reachable.
\end{proof}


\subsection*{Acknowledgements}
Research reported in this publication was supported by King Abdullah University of Science and Technology (KAUST).
\bibliographystyle{spmpsci}

\bibliography{time-space}

\begin{thebibliography}{10}
\providecommand{\url}[1]{{#1}}
\providecommand{\urlprefix}{URL }
\expandafter\ifx\csname urlstyle\endcsname\relax
  \providecommand{\doi}[1]{DOI~\discretionary{}{}{}#1}\else
  \providecommand{\doi}{DOI~\discretionary{}{}{}\begingroup
  \urlstyle{rm}\Url}\fi

\bibitem{AbouEisha19}
AbouEisha, H., Amin, T., Chikalov, I., Hussain, S., Moshkov, M.: Extensions of
  Dynamic Programming for Combinatorial Optimization and Data Mining,
  \emph{Intelligent Systems Reference Library}, vol. 146.
\newblock Springer, Cham (2019)

\bibitem{Alsolami20}
Alsolami, F., Azad, M., Chikalov, I., Moshkov, M.: Decision and Inhibitory
  Trees and Rules for Decision Tables with Many-valued Decisions,
  \emph{Intelligent Systems Reference Library}, vol. 156.
\newblock Springer, Cham (2020)

\bibitem{Ben-Or83}
Ben{-}Or, M.: Lower bounds for algebraic computation trees (preliminary
  report).
\newblock In: 15th Annual {ACM} Symposium on Theory of Computing, STOC 1983,
  pp. 80--86. {ACM}, New York, NY (1983)

\bibitem{Blum87}
Blum, M., Impagliazzo, R.: Generic oracles and oracle classes (extended
  abstract).
\newblock In: 28th Annual Symposium on Foundations of Computer Science, FOCS
  1987, pp. 118--126. {IEEE} Computer Society, Washington, DC (1987)

\bibitem{Breiman84}
Breiman, L., Friedman, J.H., Olshen, R.A., Stone, C.J.: Classification and
  Regression Trees.
\newblock Wadsworth, Belmont, CA (1984)

\bibitem{Buhrman02}
Buhrman, H., de~Wolf, R.: Complexity measures and decision tree complexity: a
  survey.
\newblock Theor. Comput. Sci. \textbf{288}(1), 21--43 (2002)

\bibitem{Dobkin76}
Dobkin, D.P., Lipton, R.J.: Multidimensional searching problems.
\newblock {SIAM} J. Comput. \textbf{5}(2), 181--186 (1976)

\bibitem{Dobkin78}
Dobkin, D.P., Lipton, R.J.: A lower bound of the $(1/2)n^2$ on linear search
  programs for the knapsack problem.
\newblock J. Comput. Syst. Sci. \textbf{16}(3), 413--417 (1978)

\bibitem{Dobkin79}
Dobkin, D.P., Lipton, R.J.: On the complexity of computations under varying
  sets of primitives.
\newblock J. Comput. Syst. Sci. \textbf{18}(1), 86--91 (1979)

\bibitem{Grigoriev95}
Grigoriev, D., Karpinski, M., Vorobjov, N.: Improved lower bound on testing
  membership to a polyhedron by algebraic decision trees.
\newblock In: 36th Annual Symposium on Foundations of Computer Science, FOCS
  1995, pp. 258--265. {IEEE} Computer Society, Washington, DC (1995)

\bibitem{Grigoriev98}
Grigoriev, D., Karpinski, M., Yao, A.C.: An exponential lower bound on the size
  of algebraic decision trees for max.
\newblock Computational Complexity \textbf{7}(3), 193--203 (1998)

\bibitem{Heide83}
{Meyer auf der Heide}, F.: A polynomial linear search algorithm for the
  $n$-dimensional knapsack problem.
\newblock In: 15th Annual {ACM} Symposium on Theory of Computing, STOC 1983,
  pp. 70--79. {ACM}, New York, NY (1983)

\bibitem{Moravek72}
Mor{\'{a}}vek, J.: A localization problem in geometry and complexity of
  discrete programming.
\newblock Kybernetika \textbf{8}(6), 498--516 (1972)

\bibitem{Moshkov82}
Moshkov, M.: On conditional tests.
\newblock Sov. Phys. Dokl. \textbf{27}, 528--530 (1982)

\bibitem{Moshkov94b}
Moshkov, M.: Decision Trees. Theory and Applications (in Russian).
\newblock Nizhny Novgorod University Publishers, Nizhny Novgorod (1994)

\bibitem{Moshkov94a}
Moshkov, M.: Optimization problems for decision trees.
\newblock Fundam. Inform. \textbf{21}(4), 391--401 (1994)

\bibitem{Moshkov95a}
Moshkov, M.: About the depth of decision trees computing {B}oolean functions.
\newblock Fundam. Inform. \textbf{22}(3), 203--215 (1995)

\bibitem{Moshkov95}
Moshkov, M.: Two approaches to investigation of deterministic and
  nondeterministic decision trees complexity.
\newblock In: 2nd World Conference on the Fundamentals of Artificial
  Intelligence, {WOCFAI} 1995, pp. 275--280 (1995)

\bibitem{Moshkov96}
Moshkov, M.: Comparative analysis of deterministic and nondeterministic
  decision tree complexity. {G}lobal approach.
\newblock Fundam. Inform. \textbf{25}(2), 201--214 (1996)

\bibitem{Moshkov00}
Moshkov, M.: On time and space complexity of deterministic and nondeterministic
  decision trees.
\newblock In: 8th International Conference Information Processing and
  Management of Uncertainty in Knowledge-based Systems, IPMU 2000, vol.~3, pp.
  1932--1936 (2000)

\bibitem{Moshkov03}
Moshkov, M.: Classification of infinite information systems depending on
  complexity of decision trees and decision rule systems.
\newblock Fundam. Inform. \textbf{54}(4), 345--368 (2003)

\bibitem{Moshkov05a}
Moshkov, M.: Comparative analysis of deterministic and nondeterministic
  decision tree complexity. {L}ocal approach.
\newblock In: Trans. Rough Sets IV, \emph{Lecture Notes in Computer Science},
  vol. 3700, pp. 125--143. Springer, Berlin Heidelberg (2005)

\bibitem{Moshkov05}
Moshkov, M.: Time complexity of decision trees.
\newblock In: Trans. Rough Sets III, \emph{Lecture Notes in Computer Science},
  vol. 3400, pp. 244--459. Springer, Berlin Heidelberg (2005)

\bibitem{Moshkov20}
Moshkov, M.: Comparative Analysis of Deterministic and Nondeterministic
  Decision Trees, \emph{Intelligent Systems Reference Library}, vol. 179.
\newblock Springer, Cham (2020)

\bibitem{Moshkov11}
Moshkov, M., Zielosko, B.: Combinatorial Machine Learning - {A} Rough Set
  Approach, \emph{Studies in Computational Intelligence}, vol. 360.
\newblock Springer, Berlin Heidelberg (2011)

\bibitem{Naiman96}
Naiman, D.Q., Wynn, H.P.: Independence number and the complexity of families of
  sets.
\newblock Discr. Math. \textbf{154}, 203--216 (1996)

\bibitem{Pawlak81}
Pawlak, Z.: Information systems theoretical foundations.
\newblock Inf. Syst. \textbf{6}(3), 205--218 (1981)

\bibitem{Rokach07}
Rokach, L., Maimon, O.: Data Mining with Decision Trees - Theory and
  Applications, \emph{Series in Machine Perception and Artificial
  Intelligence}, vol.~69.
\newblock WorldScientific, Singapore (2007)

\bibitem{Sauer72}
Sauer, N.: On the density of families of sets.
\newblock J. of Combinatorial Theory (A) \textbf{13}, 145--147 (1972)

\bibitem{Shelah72}
Shelah, S.: A combinatorial problem; stability and order for models and
  theories in infinitary languages.
\newblock Pacific J. of Mathematics \textbf{41}, 241--261 (1972)

\bibitem{Steele82}
Steele, J.M., Yao, A.C.: Lower bounds for algebraic decision trees.
\newblock J. Algorithms \textbf{3}(1), 1--8 (1982)

\bibitem{Yao92}
Yao, A.C.: Algebraic decision trees and {E}uler characteristics.
\newblock In: 33rd Annual Symposium on Foundations of Computer Science, FOCS
  1992, pp. 268--277. {IEEE} Computer Society, Washington, DC (1992)

\bibitem{Yao94}
Yao, A.C.: Decision tree complexity and {B}etti numbers.
\newblock In: 26th Annual {ACM} Symposium on Theory of Computing, STOC 1994,
  pp. 615--624. {ACM}, New York, NY (1994)

\end{thebibliography}

\end{document}